\documentclass[aps,prl,reprint]{revtex4-1}
\pdfoutput=1

\usepackage[utf8]{inputenc}
\usepackage{hyperref}
\usepackage{fullpage}
\usepackage{amssymb}
\usepackage{mathtools}
\usepackage{amsthm}
\usepackage{eucal}
\usepackage{dsfont}
\usepackage[T1]{fontenc}
\usepackage{lmodern}
\usepackage[stretch=10,shrink=10]{microtype}
\usepackage[capitalise]{cleveref}
\usepackage{color,soul}
\usepackage{tikz-cd}

\allowdisplaybreaks[1]

\theoremstyle{plain}
\newtheorem{theorem}{Theorem}[section]
\newtheorem{lemma}[theorem]{Lemma}
\newtheorem*{mylemma}{Lemma}

\theoremstyle{definition}
\newtheorem{definition}[theorem]{Definition}
\newtheorem*{mydefinition}{Definition}

\newtheorem{fact}[theorem]{Fact}

\newcommand {\eps} {\varepsilon}
\newcommand {\br} [1] {\ensuremath{ \left( #1 \right) }}

\newcommand {\minusspace} {\: \! \!}
\newcommand {\smallspace} {\: \!}
\newcommand {\fn} [2] {\ensuremath{ #1 \minusspace \br{ #2 } }}

\newcommand {\ball} [2] {\fn{\mathcal{B}^{#1}}{#2}}
\newcommand {\defeq} {\ensuremath{ \stackrel{\mathrm{def}}{=} }}
\newcommand {\mutinf} [2] {\fn{\mathrm{I}}{#1 \smallspace : \smallspace #2}}
\newcommand {\imax} [2] {\fn{\mathrm{I}_{\max}}{#1 \smallspace : \smallspace #2}}
\newcommand {\imaxeps} [2] {\fn{\mathrm{I}^{\varepsilon}_{\max}}{#1 \smallspace : \smallspace #2}}
\newcommand {\condmutinf} [3] {\mutinf{#1}{#2 \smallspace \middle\vert \smallspace #3}}

\newcommand {\norm} [1] {\ensuremath{ \left\| #1 \right\| }}
\newcommand {\normsub} [2] {\ensuremath{ \norm{#1}_{#2} }}
\newcommand {\onenorm} [1] {\normsub{#1}{1}}
\newcommand {\ent} [1] {\fn{\mathrm{S}}{#1}}
\newcommand {\relent} [2] {\fn{\mathrm{D}}{#1 \middle\| #2}}
\newcommand {\dmax} [2] {\fn{\mathrm{D}_{\max}}{#1 \middle\| #2}}
\newcommand {\hmin} [2] {\fn{\mathrm{H}_{\min}}{#1 \middle | #2}}
\newcommand {\hmineps} [2] {\fn{\mathrm{H}^{\varepsilon}_{\min}}{#1 \middle | #2}}
\newcommand {\hmax} [2] {\fn{\mathrm{H}_{\max}}{#1 \middle | #2}}
\newcommand {\hmaxeps} [2] {\fn{\mathrm{H}^{\varepsilon}_{\max}}{#1 \middle | #2}}

\newcommand {\bra} [1] {\ensuremath{ \left\langle #1 \right| }}
\newcommand {\ket} [1] {\ensuremath{ \left| #1 \right\rangle }}
\newcommand {\ketbratwo} [2] {\ensuremath{ \left| #1 \middle\rangle \middle\langle #2 \right| }}
\newcommand {\ketbra} [1] {\ketbratwo{#1}{#1}}
\newcommand {\cspace} [1] {\ensuremath{\mathnormal{#1}}}

\newcommand {\Tr} {\ensuremath{ \mathrm{Tr} }}
\newcommand {\partrace} [2] {\fn{\Tr_{#1}}{#2}}

\newcommand {\Bob} {{\mathsf{Bob}}}
\newcommand {\Alice} {{\mathsf{Alice}}}
\newcommand {\Referee} {{\mathsf{Referee}}}
\newcommand {\suppress}[1]{}

\newcommand {\set} [1] {\ensuremath{ \left\lbrace #1 \right\rbrace }}
\newcommand {\reg} [1] {\ensuremath{ \mathnormal{#1} }}

\def\P{\CMcal{P}}
\def\Q{\CMcal{Q}}

\def\H{\CMcal{H}}

\def\F{\mathrm{F}}

\def\inf{\mathrm{inf}}
\def\max{\mathrm{max}}

\def\min{\mathrm{min}}

\def\E{\mathcal{E}}

\newcommand{\comment}[1]{\textup{{\color{red}#1}}}

\newcommand {\mytitle} {Quantum message compression with applications}
\newcommand {\Rahul}   {Rahul Jain}
\newcommand {\Anurag}  {Anurag Anshu}
\newcommand {\Vamsi} {Vamsi Krishna Devabathini}

\newcommand {\CQT} {Centre for Quantum Technologies}
\newcommand {\CQTCS} {\CQT{} and Department of Computer Science}
\newcommand {\NUS} {National University of Singapore}
\newcommand {\Maju} {MajuLab, CNRS-UNS-NUS-NTU International Joint Research Unit, UMI 3654, Singapore.}

\newcommand {\authorblock} [3] {
	\begin{minipage}[t]{0.3\linewidth}
		\centering
		{#1}\\[0.8ex]
		{\footnotesize {#2}\\[-0.7ex]
		\email{#3}}
	\end{minipage}\vspace{1ex}
}

\hypersetup{
	pdfstartview={FitH},
	pdfdisplaydoctitle={true},
	breaklinks={true},
	bookmarksopen={true},
	bookmarksnumbered={false},
	pdftitle={\mytitle},
	pdfauthor={\Anurag, \Rahul, \Vamsi}
}

\begin{document}

\title{\textbf{\mytitle}\\[2ex]}

\author{
    \authorblock{\Anurag}{\CQT, \NUS}{a0109169@u.nus.edu}
	\authorblock{\Vamsi}{\CQT, \NUS}{devabathini92@gmail.com}
		\authorblock{\Rahul}{\CQTCS, \NUS \\ \Maju}{rahul@comp.nus.edu.sg}\\
}





\begin{abstract}
We present a new scheme for the compression of one-way quantum messages, in the setting of coherent entanglement assisted quantum communication. For this, we present a new technical tool that we call the {\em convex split} lemma, which is inspired by the classical compression schemes that use {\em rejection sampling} procedure. As a consequence, we show new bounds on the quantum communication cost of single-shot entanglement-assisted one-way {\em quantum state redistribution} task and for the sub-tasks {\em quantum state splitting} and {\em quantum state merging}.  Our upper and lower bounds are tight upto a constant and hence stronger than previously known best bounds for above tasks. Our protocols use explicit quantum operations on the sides of $\Alice$ and $\Bob$, which are different from the {\em decoupling by random unitaries} approach used in previous works. As another application, we present a port-based teleportation scheme which works when the set of input states is restricted to a known ensemble, hence potentially saving the number of required ports. Furthermore, in case of no prior knowledge about the set of input states, our average success fidelity matches the known average success fidelity, providing a new port-based teleportation scheme with similar performance as appears in literature. 
\end{abstract}

\maketitle

Quantum message compression is a fundamental area of quantum information theory. Schumacher\cite{Schumacher95} provided one of the first such schemes for source compression which was a direct quantum analogue of the celebrated classical source-coding scheme of Shannon \cite{Shannon}. To capture more general quantum tasks, such as those involving side information with the receiver, the work \cite{horodecki07} introduced the task of \textit{quantum state merging}. This work also provided an operational interpretation to quantum conditional entropy, the negativity of which is a well known example of peculiarities of quantum information. The task of \textit{quantum state splitting} was subsequently introduced in \cite{AbeyesingheDHW09}. A central theme in these results is the notion of a purifying system (often termed as the reference system), which brings an element of coherence in quantum protocols. Originally studied in the asymptotic and i.i.d setting, these tasks were also subsequently studied the one-shot setting \cite{Berta09,Renner11}.

The task of \textit{quantum state redistribution} elegantly captures the problem of coherent quantum message compression. In this task,
$\Alice$, $\Bob$ and $\Referee$ share a pure state $\ket{\Psi}_{RABC}$, with $AC$ belonging to $\Alice$, $B$ to $\Bob$ and $R$ to $\Referee$. $\Alice$ wants to transfer the register $C$ to $\Bob$, such that the final state $\Phi_{RABC}$ satisfies $\F(\Phi_{RABC},\Psi_{RABC})\geq \sqrt{1-\eps^2}$, for a given $\eps \geq 0$. Here $\F(.,.)$ is fidelity.

This problem has been well studied in the literature both in the asymptotic and single-shot settings (see e.g.~\cite{Devatakyard,oppenheim08,YeBW08,YardD09,Frederic10,BuscemiD10,DattaH11,DupuisBWR14,Oppenheim14} and references therein). Quantum state merging is a special case of this task when register $A$ is not present and quantum state splitting is the special case in which register $B$ is not present.  In the setting where Alice, Bob and Referee share $n$ copies of independent and identical states $\Psi_{RABC}^{\otimes n}$, it was shown by Devatak and Yard~\cite{Devatakyard,YardD09} (see also Luo and Devatak~\cite{LuoD09}) that the quantum communication cost, using one-way communication and shared-entanglement, for quantum state redistribution approaches $n\condmutinf{C}{R}{B}_{\Psi}$ as $n\rightarrow \infty$ and error $\eps\rightarrow 0$. Here, $\condmutinf{C}{R}{B}_{\Psi} = \ent{\Psi_{RB}}+\ent{\Psi_{BC}}-\ent{\Psi_B}-\ent{\Psi_{RBC}}$ is the quantum \textit{conditional mutual information} and $\ent{.}$ is the von-neumann entropy.  Subsequently, it was shown by Oppenheim~\cite{oppenheim08} that quantum state redistribution can be realized with two application of a protocol for quantum state merging. It was independently shown by Ye \textit{et.al.}~\cite{YeBW08} that quantum state redistribution can be realized with application of protocols for quantum state merging and quantum state splitting. 

Recently, in independent works by Berta, Christandl, Touchette~\cite{Berta14} and Datta, Hsieh, Oppenheim~\cite{Oppenheim14}, single-shot entanglement-assisted one-way protocols for quantum state redistribution have been proposed. The work ~\cite{Berta14} also provides several lower bounds with gaps between the upper and lower bounds and the question of closing these gaps has been left open. The upper bound of~\cite{Berta14} and~\cite{Oppenheim14} has recently been used by Touchette~\cite{Dave14} to obtain a {\em direct-sum result} for bounded-round entanglement-assisted quantum communication complexity.  

In this work, we introduce a novel technique for compressing coherent quantum information. As an immediate application, we exhibit a quantity that near-optimally captures the communication costs of quantum state re-distribution and present near optimal bounds for quantum state merging and quantum state splitting. We also give applications to the case of port- based teleportation, presenting schemes which allow port based teleportation when the states to be sent belong to a given ensemble. Our compression protocol has an important property of being explicit and using simple form of shared entanglement. Our techniques have also found recent application in the work \cite{Majenzetal}, in context of catalytic decoupling. We note that improved version of our main lemma, as presented below, also quantitatively improves one of the results in \cite{Majenzetal}.

\subsection{Preliminaries}

We represent the set of quantum states on a register $A$ with the symbol $\mathcal{D}(A)$. A subscript to a quantum state represents the register associated to it. Fidelity between states $\rho,\sigma$ is represented as $\F(\rho,\sigma) \defeq \|\sqrt{\rho}\sqrt{\sigma}\|_1$. We shall use the notation of \textit{epsilon ball}, representing $\ball{\eps}{\rho}$ as the set of all states $\sigma$ such that $\F^2(\rho,\sigma)\geq 1-\eps^2$. The \textit{relative entropy} between quantum states $\rho,\sigma$ is defined as $\relent{\rho}{\sigma} = \Tr(\rho\log\rho)-\Tr(\rho\log\sigma)$. The \textit{max-entropy} is defined as $\dmax{\rho}{\sigma}\defeq \inf\{\lambda: \rho \leq 2^{\lambda}\sigma\}$, where $A\leq B$ (for hermitian matrices $A,B$) implies that $B-A$ is a positive semidefinite matrix. The \textit{max-information} of a bipartite state $\rho_{AB}$ is defined as $\imax{A}{B}_{\rho}\defeq \inf_{\sigma\in \mathcal{D}(B)}\dmax{\rho_{AB}}{\rho_A\otimes \sigma_B}$. The \textit{smooth max-information is defined as} $\imaxeps{A}{B}_{\rho}\defeq \inf_{\rho'_{AB}\in \ball{\eps}{\rho_{AB}}}\imax{A}{B}_{\rho'}$.

\subsection{One way protocols and convex split}
\label{subsec:onewaysplit}

We begin with the following classical protocol for message compression. Alice and Referee share the mixed state $\Psi_{RA} = \sum_i p_i \ketbra{i}_R\otimes \ketbra{i}_A$ and Alice needs to send the register $A$ to Bob. A simple strategy using the technique of rejection sampling (see e.g.~\cite{HJMR10}) to achieve this task is that Alice and Bob share many copies of the state $\theta_{E_AE_B} \defeq \sum_i p_i \ketbra{i}_{E_A}\otimes \ketbra{i}_{E_B}$, and Alice sequentially checks in each copy whether the contents of registers $A$ and $E_A$ match. That is, she applies the projector $\sum_i \ketbra{i}_A\otimes \ketbra{i}_{E_A}$ and tries a fresh shared randomness upon failure. Upon success, she sends the index of the succeeding shared randomness, and Bob merely outputs his part of the shared randomness. A modification of this protocol also works when the state in register $A$ is a mixed classical/quantum state, for each $i$~\cite{HJMR10, Jain:2003, Jain:2005}. 

Unfortunately, the technique fails when the state $\Psi_{RA}$ is pure. The failure of the measurement leads to correlation between the register $R$ and parts of shared entanglement with Alice which disrupts the requirement that $R$ be correlated only with register output with Bob. To get around this problem, we design suitable operation on the Bob's side, and construct Alice's measurements in a coherent fashion. We outline our strategy more precisely below. 

Given a one-way communication protocol that achieves certain task, let the shared state between $\Alice~(A), \Bob~(B)$ and $\Referee~(R)$ be $\phi_{RAB}$. Here the registers $A,B$ denote all the registers with $\Alice$ and $\Bob$ respectively. Let a measurement $\{M_1,M_2\ldots\}$ be performed by $\Alice$. Then conditioned on an outcome $i$, state on registers $RB$ is $\Tr_{A}(M_i\phi_{RAB})$ (with slight abuse of notation for brevity). Thus, the measurement by $\Alice$ induces a \textit{convex-split} of the state $\phi_{RB}$ as follows:
$$\phi_{RB} = \sum_i \Tr_{A}(M_i\phi_{RAB}).$$ Upon receiving the message $i$ from $\Alice$, $\Bob$ `knows' that state on registers $R,B$ is $\Tr_{A}(M_i\phi_{RAB})$. They can perform further operations to finish the protocol.
   
Alternatively, given a convex-split of the state $\phi_{RB}$ as $\phi_{RB}=\sum_i p_i\phi^i_{RB}$, one can construct a one-way communication protocol as follows. Consider the following purification of $\phi_{RB}$: $$\ket{\phi}_{RBA'J}=\sum_i \sqrt{p_i}\ket{\phi^i}_{RBA'}\ket{i}_J.$$ Here $\ket{\phi^i}_{RBA'}$ is a purification of $\phi^i_{RB}$. $\Alice$ applies an isometry $V:A\rightarrow A'J$ to transform the shared state $\phi_{RAB}$ to $\phi_{RBA'J}$ and measures the register $J$. She then sends the measurement outcome to $\Bob$. Upon receiving outcome $i$, $\Bob$ `knows' that the state on registers $RB$ is $\phi^i_{RB}$, on which he can perform further operations.

We shall exploit this equivalence between convex-split and one-way communication protocols and construct a suitable convex-split corresponding to quantum state redistribution. Our idea is inspired by the aforementioned classical protocol, in which $\Bob$ simply outputs the correct register, receiving the message from Alice.

\begin{figure}[h]
\centering
\begin{tikzpicture}[xscale=0.9,yscale=1.2]



\node at (2.5,5) {$\frac{1}{n}$};

\draw[fill] (3.5,7) circle [radius=0.05];
\draw[fill] (3.5,8) circle [radius=0.05];
\draw[thick] (3.5,8) to [out=200,in=160] (3.5,7);
\node at (2.7,7.5) {$\rho_{PQ_1}$};
\node at (3.5,6.5) {$\otimes$};
\node at (2.9,6) {$\sigma_{Q_2}$};
\draw[fill] (3.5,6) circle [radius=0.05];
\node at (3.5,5.5) {$\otimes$};
\draw[fill] (3.5,5) circle [radius=0.05];
\node at (3.5,4.5) {$\otimes$};
\draw[fill] (3.5,4.1) circle [radius=0.02];
\draw[fill] (3.5,3.7) circle [radius=0.02];
\draw[fill] (3.5,3.3) circle [radius=0.02];
\draw[fill] (3.5,2.9) circle [radius=0.02];
\node at (3.5,2.5) {$\otimes$};
\draw[fill] (3.5,2) circle [radius=0.05];
\node at (2.9,2) {$\sigma_{Q_n}$};

\node at (4.5,5) {$+$};

\node at (5.5,5) {$\frac{1}{n}$};

\draw[fill] (6.5,7) circle [radius=0.05];
\draw[fill] (6.5,8) circle [radius=0.05];
\node at (7.1,7) {$\sigma_{Q_1}$};
\draw[thick] (6.5,8) to [out=200,in=160] (6.5,6);
\node at (5.5,7) {$\rho_{PQ_2}$};
\node at (6.5,6.5) {$\otimes$};
\draw[fill] (6.5,6) circle [radius=0.05];
\node at (6.5,5.5) {$\otimes$};
\draw[fill] (6.5,5) circle [radius=0.05];
\node at (6.5,4.5) {$\otimes$};
\draw[fill] (6.5,4.1) circle [radius=0.02];
\draw[fill] (6.5,3.7) circle [radius=0.02];
\draw[fill] (6.5,3.3) circle [radius=0.02];
\draw[fill] (6.5,2.9) circle [radius=0.02];
\node at (6.5,2.5) {$\otimes$};
\draw[fill] (6.5,2) circle [radius=0.05];
\node at (5.9,2) {$\sigma_{Q_n}$};

\draw[fill] (7.5,5.0) circle [radius=0.02];
\draw[fill] (7.9,5.0) circle [radius=0.02];
\draw[fill] (8.3,5.0) circle [radius=0.02];
\draw[fill] (8.7,5.0) circle [radius=0.02];

\node at (9.2,5) {$+$};

\node at (9.7,5) {$\frac{1}{n}$};

\draw[fill] (10.5,7) circle [radius=0.05];
\draw[fill] (10.5,8) circle [radius=0.05];
\draw[thick] (10.5,8) to [out=230,in=85] (10.1,7) to [out=265,in=95] (10.1,3) to [out=275,in=130] (10.5,2);
\node at (9.7,7.5) {$\rho_{PQ_n}$};
\node at (11.1,7) {$\sigma_{Q_1}$};
\node at (10.5,6.5) {$\otimes$};
\node at (11.1,6) {$\sigma_{Q_2}$};
\draw[fill] (10.5,6) circle [radius=0.05];
\node at (10.5,5.5) {$\otimes$};
\draw[fill] (10.5,5) circle [radius=0.05];
\node at (10.5,4.5) {$\otimes$};
\draw[fill] (10.5,4.1) circle [radius=0.02];
\draw[fill] (10.5,3.7) circle [radius=0.02];
\draw[fill] (10.5,3.3) circle [radius=0.02];
\draw[fill] (10.5,2.9) circle [radius=0.02];
\node at (10.5,2.5) {$\otimes$};
\draw[fill] (10.5,2) circle [radius=0.05];

\end{tikzpicture}
\caption{The state $\tau_{PQ_1Q_2\ldots Q_n}$ considered in convex-split lemma. For large enough $n$, $\tau_{PQ_1Q_2\ldots Q_n}$ is approximately equal to the state $\tau_P\otimes\tau_{Q_1Q_2\ldots Q_n}$.}
 \label{fig:convex-split}
\end{figure}
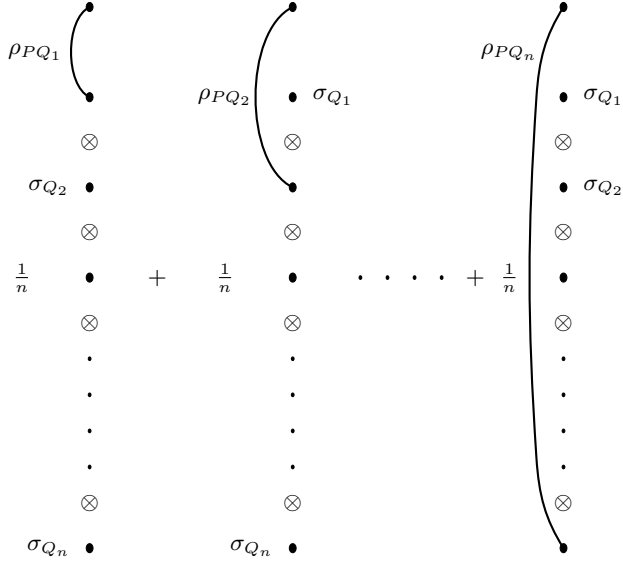

We shall show the following main lemma, which we apply to compression of quantum messages in next section.  

\begin{mylemma}[Convex-split lemma]
Let $\rho_{PQ}\in\mathcal{D}(PQ)$ and $\sigma_Q\in\mathcal{D}(Q)$ be quantum states such that $\text{supp}(\rho_Q)\subset\text{supp}(\sigma_Q)$.  Let $k \defeq \dmax{\rho_{PQ}}{\rho_P\otimes\sigma_Q}$. Define the following state (please also refer to  Figure~\ref{fig:convex-split})
\begin{eqnarray} \label{eq:convexsplit}
&& \tau_{PQ_1Q_2\ldots Q_n} \defeq  \frac{1}{n}\sum_{j=1}^n \rho_{PQ_j}\otimes\nonumber\\ &&\sigma_{Q_1}\otimes \sigma_{Q_2}\ldots\otimes\sigma_{Q_{j-1}}\otimes\sigma_{Q_{j+1}}\ldots\otimes\sigma_{Q_n}
\end{eqnarray}
on $n+1$ registers $P,Q_1,Q_2,\ldots Q_n$, where $\forall j \in [n]: \rho_{PQ_j} = \rho_{PQ}$ and $\sigma_{Q_j}=\sigma_Q$.  Then,  
$$ \F^{2}(\tau_{PQ_1Q_2\ldots Q_n},\rho_P \otimes \sigma_{Q_1}\otimes\sigma_{Q_2}\ldots \otimes \sigma_{Q_n}) \geq \frac{1}{1+\frac{2^k}{n}}.$$ In particular, for $\delta \in (0,1/3) $ and $ n\defeq \lceil\frac{2^k}{\delta}\rceil$, 
$$ \F^{2}(\tau_{PQ_1Q_2\ldots Q_n},\rho_P \otimes \sigma_{Q_1}\otimes\sigma_{Q_2}\ldots \otimes \sigma_{Q_n}) \geq 1 - \delta.$$ 
\end{mylemma}

\subsection{Compressing one-way quantum message}
Using the convex split lemma, we now present a scheme to compress an arbitrary quantum message.  Consider a protocol in which $\Alice(AM)$, $\Bob(B)$ and $\Referee(R)$ share a joint quantum state $\Phi_{RAMB}$ and Alice sends the register $M$ to $\Bob$. This protocol may appear as a sub-routine in any other quantum protocol. We shall show the following theorem.

\begin{theorem}
\label{thm:compression}
There exists a protocol $\P$ which starts with the state $\Phi_{RAMB}$ and produces a state $\Phi'_{RAMB}$ with the property $\F^2(\Phi_{RAMB},\Phi'_{RAMB})\geq 1-\eps^2$, such that the register $M$ is held by $\Bob$. The number of qubits communicated from $\Alice$ to $\Bob$ is at most $$\frac{1}{2}\imax{RB}{M}_{\Phi} + \log\left(\frac{1}{\eps}\right) .$$
\end{theorem} 

An outline of the proof is as follows. The compression of message register $M$ amounts to constructing a suitable convex split on the registers involved with $\Bob$ and $\Referee$.

For any $\sigma\in \mathcal{D}(M)$, we set $\delta=\eps^2$, $k= \dmax{\Phi_{RBM}}{\Phi_{RB}\otimes \sigma_M}$, $n=2^k/\delta$ and consider the state 
\begin{eqnarray*}
&& \tau_{RBM_1M_2\ldots M_n} \defeq \nonumber\\ && \frac{1}{n}\sum_{j=1}^n \Phi_{RBM_j}\otimes\sigma_{M_1}\ldots\otimes\sigma_{M_{j-1}}\otimes\sigma_{M_{j+1}}\ldots\otimes\sigma_{M_n}
\end{eqnarray*}

Suppose $\Alice$, $\Bob$ and $\Referee$ hold the following purification of  $\tau_{RBM_1M_2\ldots M_n}$, 
\begin{eqnarray*}
&& \ket{\tau}_{RBJL_1\ldots L_nM_1\ldots M_n} \defeq \frac{1}{\sqrt{n}}\sum_{j=1}^n \ket{\tilde{\Phi}}_{RBL_jM_j} \otimes\ket{j}_J\otimes\\ && \ket{\sigma}_{L_1M_1}\ldots\otimes\ket{\sigma}_{L_{j-1}M_{j-1}}\otimes\ket{\sigma}_{L_{j+1}M_{j+1}}\ldots\otimes\ket{\sigma}_{L_nM_n}
\end{eqnarray*}
where, $L_i$ is a register purifying $M_i$ and $\ket{\tilde{\Phi}}_{RBL_jM_j}$ is a purification of $\Phi_{RBM_j}$. Then $\Alice$ can send the register $J$ to $\Bob$ via superdense coding. Upon receiving the message $j$, Bob can pick up the register $M_j$ and Alice can purify the state $\Phi_{RBM_j}$ in her register $A$. The state in registers $RBM_jA$ is then the desired output.

Obvious issue with this scheme is that the parties do not possess the state $\tau_{RBAJL_1L_2\ldots L_nM_1M_2\ldots M_n}$. But observe that $\tau_{RB} = \Phi_{RB}$ and  $\tau_{RB} \otimes \sigma_{M_1}\otimes\sigma_{M_2}\ldots \otimes \sigma_{M_n}$ and $\tau_{RBM_1M_2\ldots M_n}$ are $2\delta$-close in fidelity (due to the convex split lemma). Thus, the parties can start with the state $\Phi_{RBAM}$ and the purification $\ket{\sigma}_{L_1M_1}\otimes \ldots \ket{\sigma}_{L_nM_n}$ of $ \sigma_{M_1}\otimes\sigma_{M_2}\ldots \otimes \sigma_{M_n}$ as the pre-shared entanglement and $\Alice$ can apply an isometry $V$ on her registers (guaranteed by Uhlmann's theorem \cite{uhlmann76}) such that the states $V(\ket{\Phi}_{RBAM} \otimes \ket{\sigma}_{L_1M_1}\otimes \ldots \ket{\sigma}_{L_nM_n})$ and $\ket{\tau'}_{RBAJL_1L_2\ldots L_nM_1M_2\ldots M_n}$ are $\delta$-close to each other in fidelity. $\Alice$ then sends the register $J$, leading to a state $\Phi'_{RBM_jA}$ shared between all three parties, such that $\F^2(\Phi'_{RBM_jA},\Phi_{RBM_jA})\geq 1-2\delta$. The number of qubits communicated by $\Alice$ is  
$\frac{\log(n)}{2}$. The theorem follows when we choose $\sigma$ such that $\dmax{\Phi_{RBM}}{\Phi_{RB}\otimes \sigma_M} = \imax{RB}{M}_{\Phi}$.

\subsection{Application: Quantum state redistribution}

An immediate application of our compression result is near optimal characterization of the task of quantum state redistribution. We begin with the case for quantum state splitting (in which register $B$ is trivial).  $\Referee(R)$ and $\Alice(AC)$ share the state $\Psi_{RAC}$ and $\Alice$ needs to send the register $C$ to $\Bob$. From our compression result, a protocol in which $\Alice$ simply sends the register $C$ to $\Bob$ can be compressed using a new protocol $\P$ which makes an error of at most $2\eps$, using the following number of qubits: 
$$\frac{1}{2}\imaxeps{R}{C}_{\Psi} + \log\left(\frac{2}{\eps}\right) .$$

It is known that any one-way entanglement assisted quantum protocol that makes an error at most $\eps$ must communicate at least $\frac{1}{2}\imaxeps{R}{C}_{\Psi}$ number of qubits \cite{Renner11}. Similar bounds also hold for quantum state merging (in which register $A$ is trivial), as quantum state merging can be viewed as a time-reversed version of quantum state splitting \cite{Renner11}. A slightly weaker form of our result was already known in \cite{Renner11}, where the protocol used $\frac{1}{2}\imaxeps{R}{C}_{\Psi} + \log\log\mathrm{dim}(C) + \mathcal{O}(\log\left(\frac{1}{\eps}\right)) $ qubits of communication and embezzling quantum states as pre-shared entanglement. On the other hand, the communication cost in our protocol differs from the lower bound by a constant and uses a simple form of pre-shared entanglement: sufficiently many copies of purifications of certain state $\sigma_C$, that is obtained in the optimization in $\imaxeps{R}{C}_{\Psi}$. 
 
We show a similar statement for the case of quantum state redistribution, where the communication cost is tighly characterized by the following quantity. 
\begin{mydefinition} Let $\varepsilon \geq 0$ and $\ket{\Psi}_{RABC}$ be a pure state. Define,
\begin{eqnarray*}  
&&\mathrm{Q}^{\varepsilon}_{\ket{\Psi}_{RABC}} \defeq \inf_{T, U_{BCT},\sigma_T,\kappa_{RBCT}} \imax{RB}{CT}_{\kappa_{RBCT}}  
\end{eqnarray*}
with the condition that $U_{BCT}$ is a unitary on registers $BCT$, $\sigma_T\in \mathcal{D}(T)$, $\kappa_{RB}=\Psi_{RB}$ and
$$(I_R \otimes U_{BCT}) \kappa_{RBCT}(I_R \otimes U^{\dagger}_{BCT}) \in \ball{\eps}{\Psi_{RBC}\otimes\sigma_T}.$$ 
\end{mydefinition}

It may be noted that in above quantity, we are not optimizing over all possible protocols. Rather, it quantifies (in terms of max-information) how well $\Bob$ can decouple the registers $RB$ and $CT$ using local operations and additional ancilla register $T$. In the special case where $B$ is trivial (for quantum state splitting), there is no additional ancillary register $T$ required by $\Bob$ for best possible decoupling and above quantity coincides with $\imaxeps{R}{C}_{\Psi}$.

We show that any one-way entanglement assisted quantum protocol achieving quantum state redistribution of the state $\Psi_{RABC}$ with error at most $\eps$ must communicate at least $\frac{1}{2}\mathrm{Q}^{\varepsilon}_{\ket{\Psi}_{RABC}}$ number of qubits. Furthemore, using our compression result, we exhibit a protocol that makes an error of at most $2\eps$ and communicates at most $\frac{1}{2}\mathrm{Q}^{\varepsilon}_{\ket{\Psi}_{RABC}} + \log\left(\frac{2}{\eps}\right) $ number of qubits. We leave further understanding of this quantity, to future work.

\subsection{Application: Port-based teleportation}
The work \cite{Portbased08} introduced the technique of Port-based teleportation, where $\Alice$ and $\Bob$ share many copies of maximally entangled states (called \textit{ports}), and upon receiving message from $\Alice$ (which she prepares after her local quantum operation), $\Bob$ simply picks up the desired state in one of the ports. It was shown in \cite{Portbased09} that there is a Port-based teleportation scheme for $d$-dimensional quantum states that uses $N$ copies of $d$-dimensional maximally entangled states and achieves average squared-fidelity of transmission at least $1- \frac{d^2}{N}$. 

Using the convex split lemma, we provide a scheme for Port-based teleportation in the presence of side information about the set of input states. Given a collection of states on a register $M$ and associated probabilities $\{p_i,\psi^i_M\}_i$, define the state $\ket{\Psi}_{RM} \defeq \sum_i \sqrt{p_i}\ket{i}_R\ket{\psi^i}_M$ (where $R$ is a register with sufficient large dimension). Then we present a protocol where $\Alice$ and $\Bob$ share $N$ copies of purification of an arbitrary state $\sigma_M$ and perform a port-based teleportation protocol for which the average squared-fidelity is at least $1- \left( 2^{\dmax{\Psi_{RM}}{\Psi_R\otimes\sigma_M}}/N \right)$. In particular, if $\psi^i_M$ form the set of all possible quantum states on $d$ dimensional register $M$ with uniform distribution, then choosing $\sigma_M= \frac{\mathrm{I}_M}{d}$, we find that $\dmax{\Psi_{RM}}{\Psi_R\otimes\sigma_M} = 2\log(d)$. Thus, this scheme achieves the average squared-fidelity at least $1-\frac{d^2}{N}$ for port-based teleportation of an arbitrary state, similar to the result obtained in \cite{Portbased09}. But one can obtain better average squared-fidelity for a different collection $\{p_i,\psi^i_M\}_i$, since it always holds that $\inf_{\sigma_M}\dmax{\Psi_{RM}}{\Psi_R\otimes\sigma_M} \leq 2\log(d)$ (as shown in \cite{Renner11}). 

\subsection{Conclusion}

We have presented a new protocol for compressing one-way coherent quantum messages upto the {\em max information} between the message and joint system between receiver and reference. As a consequence, we have obtained optimal quantities characterizing the quantum communication cost of quantum state merging and quantum state splitting tasks. We have also exhibited a similar quantity for the task of quantum state redistribution, although it remains to better understand the optimization in it. We have also exhibited a port-based teleportation scheme that can potentially save the number of ports in presence of further information about the ensemble of the states to be teleported. 
  
\section*{Acknowledgment} 
We thank Ashwin Nayak, Nilanjana Datta, Marco Tomamichel and Mark~M. Wilde for helpful comments on previous versions of this paper and also for pointing us to many useful references. We thank Debbie Leung for pointing out the connection with Port-based teleportation. We thank Priyanka Mukhopadhyay and Naqueeb Warsi for helpful discussions. This work is supported by the Singapore Ministry of Education and the National Research Foundation, also through the Tier 3 Grant ``Random numbers from quantum processes'' MOE2012-T3-1-009. 

\bibliographystyle{apsrev4-1}
\bibliography{state-redistribution}
\newpage
\appendix

\onecolumngrid
\section{Preliminaries}
\label{sec:Preliminaries}

In this section we present some notations, definitions, facts and lemmas that we will use later in our proofs. Readers may refer to~\cite{CoverT91,NielsenC00,Watrouslecturenote} for good introduction to classical and quantum information theory.

\subsection*{Information theory}

Consider a finite dimensional Hilbert space $\H$ endowed with an inner product $\langle \cdot, \cdot \rangle$ (in this paper, we only consider finite dimensional Hilbert-spaces). The $\ell_1$ norm of an operator $X$ on $\H$ is $\onenorm{X}\defeq\Tr\sqrt{X^{\dag}X}$ and $\ell_2$ norm is $\norm{X}_2\defeq\sqrt{\Tr XX^{\dag}}$. A quantum state (or a density matrix or a state) is a positive semi-definite matrix on $\H$ with trace equal to $1$. It is called {\em pure} if and only if its rank is $1$. A sub-normalized state is a positive semi-definite matrix on $\H$ with trace less than or equal to $1$. Let $\ket{\psi}$ be a unit vector on $\H$, that is $\langle \psi,\psi \rangle=1$.  With some abuse of notation, we use $\psi$ to represent the state and also the density matrix $\ketbra{\psi}$, associated with $\ket{\psi}$. Given a quantum state $\rho$ on $\H$, {\em support of $\rho$}, called $\text{supp}(\rho)$ is the subspace of $\H$ spanned by all eigen-vectors of $\rho$ with non-zero eigenvalues.
 
A {\em quantum register} $A$ is associated with some Hilbert space $\H_A$. Define $|A| \defeq \dim(\H_A)$. Let $\mathcal{L}(A)$ represent the set of all linear operators on $\H_A$. We denote by $\mathcal{D}(A)$, the set of quantum states on the Hilbert space $\H_A$. State $\rho$ with subscript $A$ indicates $\rho_A \in \mathcal{D}(A)$. If two registers $A,B$ are associated with the same Hilbert space, we shall represent the relation by $A\equiv B$.  Composition of two registers $A$ and $B$, denoted $AB$, is associated with Hilbert space $\H_A \otimes \H_B$.  For two quantum states $\rho\in \mathcal{D}(A)$ and $\sigma\in \mathcal{D}(B)$, $\rho\otimes\sigma \in \mathcal{D}(AB)$ represents the tensor product (Kronecker product) of $\rho$ and $\sigma$. The identity operator on $\H_A$ (and associated register $A$) is denoted $I_A$. 

Let $\rho_{AB} \in \mathcal{D}(AB)$. We define
\[ \rho_{\reg{B}} \defeq \partrace{\reg{A}}{\rho_{AB}}
\defeq \sum_i (\bra{i} \otimes I_{\cspace{B}})
\rho_{AB} (\ket{i} \otimes I_{\cspace{B}}) , \]
where $\set{\ket{i}}_i$ is an orthonormal basis for the Hilbert space $\H_A$.
The state $\rho_B\in \mathcal{D}(B)$ is referred to as the marginal state of $\rho_{AB}$. Unless otherwise stated, a missing register from subscript in a state will represent partial trace over that register. Given a $\rho_A\in\mathcal{D}(A)$, a {\em purification} of $\rho_A$ is a pure state $\rho_{AB}\in \mathcal{D}(AB)$ such that $\partrace{\reg{B}}{\rho_{AB}}=\rho_A$. Purification of a quantum state is not unique.

A quantum {map} $\E: \mathcal{L}(A)\rightarrow \mathcal{L}(B)$ is a completely positive and trace preserving (CPTP) linear map (mapping states in $\mathcal{D}(A)$ to states in $\mathcal{D}(B)$). A {\em unitary} operator $U_A:\H_A \rightarrow \H_A$ is such that $U_A^{\dagger}U_A = U_A U_A^{\dagger} = I_A$. An {\em isometry}  $V:\H_A \rightarrow \H_B$ is such that $V^{\dagger}V = I_A$ and $VV^{\dagger} = I_B$. The set of all unitary operations on register $A$ is  denoted by $\mathcal{U}(A)$.

\begin{definition}
We shall consider the following information theoretic quantities. Reader is referred to ~\cite{Renner05, Tomamichel09,Tomamichel12,Datta09} for many of these definitions. We consider only normalized states in the definitions below. Let $\varepsilon \geq 0$. 
\begin{enumerate}
\item {\bf Fidelity} For $\rho_A,\sigma_A \in \mathcal{D}(A)$, $$\F(\rho_A,\sigma_A)\defeq\onenorm{\sqrt{\rho_A}\sqrt{\sigma_A}}.$$ For classical probability distributions $P = \{p_i\}, Q =\{q_i\}$, $$\F(P,Q)\defeq \sum_i \sqrt{p_i \cdot q_i}.$$
\item {\bf Purified distance} For $\rho_A,\sigma_A \in \mathcal{D}(A)$, $$\P(\rho_A,\sigma_A) = \sqrt{1-\F^2(\rho_A,\sigma_A)}.$$
\item {\bf $\varepsilon$-ball} For $\rho_A\in \mathcal{D}(A)$, $$\ball{\eps}{\rho_A} \defeq \{\rho'_A\in \mathcal{D}(A)|~\P(\rho_A,\rho'_A) \leq \varepsilon\}. $$ 

\item {\bf Von-neumann entropy} For $\rho_A\in\mathcal{D}(A)$, $$\ent{\rho_A} \defeq - \Tr(\rho_A\log\rho_A) .$$ 
\item {\bf Relative entropy} For $\rho_A,\sigma_A\in \mathcal{D}(A)$ such that $\text{supp}(\rho_A) \subset \text{supp}(\sigma_A)$, $$\relent{\rho_A}{\sigma_A} \defeq \Tr(\rho_A\log\rho_A) - \Tr(\rho_A\log\sigma_A) .$$ 
\item {\bf Max-relative entropy} For $\rho_A,\sigma_A\in \mathcal{D}(A)$ such that $\text{supp}(\rho_A) \subset \text{supp}(\sigma_A)$, $$ \dmax{\rho_A}{\sigma_A}  \defeq  \inf \{ \lambda \in \mathbb{R} : 2^{\lambda} \sigma_A \geq \rho_A \}  .$$ 
\item {\bf Mutual information} For $\rho_{AB}\in \mathcal{D}(AB)$, $$\mutinf{A}{B}_{\rho}\defeq \ent{\rho_A} + \ent{\rho_B}-\ent{\rho_{AB}} = \relent{\rho_{AB}}{\rho_A\otimes\rho_B}.$$
\item {\bf Conditional mutual information} For $\rho_{ABC}\in\mathcal{D}(ABC)$, $$\condmutinf{A}{B}{C}_{\rho}\defeq \mutinf{A}{BC}_{\rho}-\mutinf{A}{C}_{\rho}.$$
\item {\bf Max-information}  For $\rho_{AB}\in \mathcal{D}(AB)$, $$ \imax{A}{B}_{\rho} \defeq   \inf_{\sigma_{B}\in \mathcal{D}(B)}\dmax{\rho_{AB}}{\rho_{A}\otimes\sigma_{B}} .$$
\item {\bf Smooth max-information} For $\rho_{AB}\in \mathcal{D}(AB)$,  $$\imaxeps{A}{B}_{\rho} \defeq \inf_{\rho'\in \ball{\eps}{\rho}} \imax{A}{B}_{\rho'} .$$	
\item {\bf Conditional min-entropy}  For $\rho_{AB}\in \mathcal{D}(AB)$,  $$ \hmin{A}{B}_{\rho} \defeq  - \inf_{\sigma_B\in \mathcal{D}(B)}\dmax{\rho_{AB}}{I_{A}\otimes\sigma_{B}} .$$  	
\suppress{
\item {\bf Conditional max-entropy}  For $\rho_{AB}\in \mathcal{D}(AB)$,  $$\hmax{A}{B}_{\rho_{AB}} \defeq - \hmin{A}{R}_{\rho_{AR}}, $$
where $\rho_{ABR}$ is a  purification of $\rho_{AB}$ for some register $R$. 
\item {\bf Smooth conditional min-entropy}  For $\rho_{AB}\in \mathcal{D}(AB)$,  $$\hmineps{A}{B}_{\rho} \defeq   \sup_{\rho^{'} \in \ball{\eps}{\rho}} \hmin{A}{B}_{\rho^{'}} .$$  	
\item {\bf Smooth conditional max-entropy}  For $\rho_{AB}\in \mathcal{D}(AB)$, $$ \hmaxeps{A}{B}_{\rho} \defeq \inf_{\rho^{'} \in \ball{\eps}{\rho}} \hmax{A}{B}_{\rho^{'}} .$$ 
}
\end{enumerate}
\label{def:infquant}
\end{definition}	
We will use the following facts. 
\begin{fact}[Triangle inequality for purified distance,~\cite{Tomamichel12}]
\label{fact:trianglepurified}
For states $\rho_A, \sigma_A, \tau_A\in \mathcal{D}(A)$,
$$\P(\rho_A,\sigma_A) \leq \P(\rho_A,\tau_A)  + \P(\tau_A,\sigma_A) . $$ 
\end{fact}

\begin{fact}[\cite{stinespring55}](\textbf{Stinespring representation})\label{stinespring}
Let $\E(\cdot): \mathcal{L}(A)\rightarrow \mathcal{L}(B)$ be a quantum operation. There exists a register $C$ and an unitary $U\in \mathcal{U}(ABC)$ such that $\E(\omega)=\Tr_{A,C}\br{U (\omega  \otimes \ketbra{0}^{B,C}) U^{\dagger}}$. Stinespring representation for a channel is not unique. 
\end{fact}

\begin{fact}[Monotonicity under quantum operations, \cite{barnum96},\cite{lindblad75}]
	\label{fact:monotonequantumoperation}
For quantum states $\rho$, $\sigma \in \mathcal{D}(A)$, and quantum operation $\E(\cdot):\mathcal{L}(A)\rightarrow \mathcal{L}(B)$, it holds that
\begin{align*}
	\onenorm{\E(\rho) - \E(\sigma)} \leq \onenorm{\rho - \sigma} \quad \mbox{and} \quad \F(\E(\rho),\E(\sigma)) \geq \F(\rho,\sigma) \quad \mbox{and} \quad \relent{\rho}{\sigma}\geq \relent{\E(\rho)}{\E(\sigma)}.
\end{align*}
In particular, for bipartite states $\rho_{AB},\sigma_{AB}\in \mathcal{D}(AB)$, it holds that
\begin{align*}
	\onenorm{\rho_{AB} - \sigma_{AB}} \geq \onenorm{\rho_A - \sigma_A} \quad \mbox{and} \quad \F(\rho_{AB},\sigma_{AB}) \leq \F(\rho_A,\sigma_A) \quad \mbox{and} \quad \relent{\rho_{AB}}{\sigma_{AB}}\geq \relent{\rho_A}{\sigma_A} .
\end{align*}
\end{fact}

\begin{fact}[Uhlmann's theorem, \cite{uhlmann76}]
\label{uhlmann}
Let $\rho_A,\sigma_A\in \mathcal{D}(A)$. Let $\rho_{AB}\in \mathcal{D}(AB)$ be a purification of $\rho_A$ and $\sigma_{AC}\in\mathcal{D}(AC)$ be a purification of $\sigma_A$. There exists an isometry $V: \H_C \rightarrow \H_B$ such that,
 $$\F(\ketbra{\theta}_{AB}, \ketbra{\rho}_{AB}) = \F(\rho_A,\sigma_A) ,$$
 where $\ket{\theta}_{AB} = (I_A \otimes V) \ket{\sigma}_{AC}$.
\end{fact}

\suppress{
\begin{fact}
\label{fact:fidelitydmax}
For quantum states $\rho_A,\sigma_A,\tau_A\in \mathcal{D}(A)$ such that $\dmax{\rho_A}{\sigma_A}=2k$, $\F(\tau_A,\rho_A)\leq 2^k\F(\tau_A,\sigma_A)$.
\end{fact}
\begin{proof} \comment{Needs more explanations.}
Using $\rho_A\leq 2^{2k}\sigma_A$, we obtain $$\F(\tau_A,\rho_A)=\|\sqrt{\tau_A}\sqrt{\rho_A}\|_1\leq 2^k\|\sqrt{\tau_A}\sqrt{\sigma_A}\|_1=2^k\F(\tau_A,\sigma_A).$$
\end{proof}
\begin{fact}
\label{fact:dmaxmorethand}
For quantum states $\rho_A, \sigma_A\in \mathcal{D}(A)$,
$$\dmax{\rho_A}{\sigma_A} \geq \relent{\rho_A}{\sigma_A}.$$
\end{fact}
\begin{proof} \comment{Mention log is operator monotone somewhere.}
Let $\dmax{\rho_A}{\sigma_A}=k$. Then $\sigma_A\geq 2^{-k}\rho_A$, which means $$\relent{\rho_A}{\sigma_A}=\Tr(\rho_A\log(\rho_A)-\rho_A\log(\sigma_A))\leq k + \Tr(\rho_A\log(\rho_A)-\rho_A\log(\rho_A))=k.$$
\end{proof}

\begin{fact}
\label{mutinfmin}
For a quantum state $\rho_{AB}\in\mathcal{D}(AB)$,
$$\mutinf{A}{B}_{\rho} =  \inf_{\sigma_B} \relent{\rho_{AB}}{\rho_A\otimes\sigma_B} = \relent{\rho_{AB}}{\rho_A\otimes\rho_B} .$$
\end{fact}	
\begin{proof}
Consider 
\begin{eqnarray*}
\relent{\rho_{AB}}{\rho_A\otimes\sigma_B}&=&\Tr(\rho_{AB}\log(\rho_{AB})-\rho_{AB}\log(\rho_A\otimes\sigma_B))\\&=&-\ent{\rho_{AB}}+\ent{\rho_A}-\Tr(\rho_B\log(\sigma_B))=\mutinf{A}{B}_{\rho}+\relent{\rho_B}{\sigma_B}.
\end{eqnarray*}
 Fact follows since $\relent{\rho_B}{\sigma_B}\geq 0$ and is $0$ iff $\rho_B=\sigma_B$. 
\end{proof}
}

\begin{fact}[\cite{Renner11}, Lemma B.7]
\label{dmaxuppbound}
For a quantum state $\rho_{AB}\in\mathcal{D}(AB)$, 
$$\imax{A}{B}_{\rho}\leq 2 \cdot \text{min}\{ \log|A|,  \log|B|\}.$$
\end{fact}
\begin{fact}[\cite{Renner11}, Lemma B.14]
\label{imaxmonotone}
For a quantum state $\rho_{ABC}\in\mathcal{D}(ABC)$,
$$\imax{A}{BC}_{\rho}\geq \imax{A}{B}_{\rho}.$$
\end{fact}

\begin{fact}[Pinsker's inequality, \cite{CastelleHR78}]
\label{pinsker}
For quantum states $\rho_A,\sigma_A\in\mathcal{D}(A)$, 
$$\F(\rho,\sigma) \geq 2^{-\frac{1}{2}\relent{\rho}{\sigma}}.$$
 This implies,
$$1 - \F(\rho,\sigma) \leq \frac{\ln 2}{2} \cdot \relent{\rho}{\sigma} \leq  \relent{\rho}{\sigma}  .$$
\end{fact}	
\begin{lemma}
\label{lem:pureclose}
Let $\epsilon>0$. Let $\ketbra{\psi}_{A}\in\mathcal{D}(A)$ be a pure state and let $\rho_{AB}\in\mathcal{D}(AB)$ be a state such that $\F(\ketbra{\psi}_A, \rho_A) \geq 1 - \varepsilon$. There exists a state $\theta_B\in\mathcal{D}(B)$ such that $\F(\ketbra{\psi}_A \otimes \theta_B, \rho_{AB}) \geq 1 - \varepsilon$.
\end{lemma}
\begin{proof}
Introduce a register $C$ such that $|C|=|A||B|$. Let $\ket{\rho}_{ABC}\in\mathcal{D}(ABC)$ be a purification of $\rho_{AB}$. Using Uhlmann's theorem (Fact~\ref{uhlmann}) we get a pure state $\theta_{BC}$ such that 
\begin{align*}
1 - \varepsilon &\leq \F(\ketbra{\psi}_A, \rho_{A}) \\
&= \F(\ketbra{\psi}_A \otimes \ketbra{\theta}_{BC}, \ketbra{\rho}_{ABC})  \\
&\leq \F(\ketbra{\psi}_A \otimes \theta_{B}, \rho_{AB}) . \quad \mbox{(monotonicity of fidelity under quantum operation, Fact~\ref{fact:monotonequantumoperation})} 
\end{align*} 
\end{proof}

The following lemma is a tighter version of (one-sided) convexity of relative entropy.

\begin{lemma}
\label{relentconcav}
Let $\mu_1,\mu_2,\ldots \mu_n, \theta$ be quantum states and $\{p_1,p_2,\ldots p_n\}$ be a probability distribution. Let $\mu=\sum_i p_i\mu_i$ be the average state. Then 

$$\relent{\mu}{\theta} =\sum_i p_i(\relent{\mu_i}{\theta}-\relent{\mu_i}{\mu}).$$
\end{lemma}
\begin{proof}
Proof proceeds by direct calculation. Consider
\begin{eqnarray*}
&&\sum_i p_i(\relent{\mu_i}{\theta}-\relent{\mu_i}{\mu}) = \sum_i p_i(\Tr(\mu_i\log\mu_i)-\Tr(\mu_i\log\theta) - \Tr(\mu_i\log\mu_i) + \Tr(\mu_i\log\mu))\\ && = \Tr(\sum_ip_i\mu_i\log(\mu)) - \Tr(\sum_ip_i\mu_i\log\theta)= \Tr(\mu\log\mu) - \Tr(\mu\log\theta) = \relent{\mu}{\theta} .
\end{eqnarray*}
\end{proof}

\section{A convex-split lemma}

\label{sec:convexcomb}
We revisit the statement of convex split lemma and state its connection to a previous work. The lemma has been proved in main text. 
\begin{lemma}[Convex-split lemma]
\label{convexcomb}
Let $\rho_{PQ}\in\mathcal{D}(PQ)$ and $\sigma_Q\in\mathcal{D}(Q)$ be quantum states such that $\text{supp}(\rho_Q)\subset\text{supp}(\sigma_Q)$.  Let $k \defeq \dmax{\rho_{PQ}}{\rho_P\otimes\sigma_Q}$. Define the following state
\begin{equation}
\tau_{PQ_1Q_2\ldots Q_n} \defeq  \frac{1}{n}\sum_{j=1}^n \rho_{PQ_j}\otimes\sigma_{Q_1}\otimes \sigma_{Q_2}\ldots\otimes\sigma_{Q_{j-1}}\otimes\sigma_{Q_{j+1}}\ldots\otimes\sigma_{Q_n}
\end{equation}
on $n+1$ registers $P,Q_1,Q_2,\ldots Q_n$, where $\forall j \in [n]: \rho_{PQ_j} = \rho_{PQ}$ and $\sigma_{Q_j}=\sigma_Q$.  Then, $$\relent{\tau_{PQ_1Q_2\ldots Q_n}}{\tau_P \otimes \sigma_{Q_1}\otimes\sigma_{Q_2}\ldots \otimes \sigma_{Q_n}} \leq \log(1+\frac{2^k}{n}).$$ Using Pinsker's inequality (Fact \ref{pinsker}), we conclude, $$ \F^{2}(\tau_{PQ_1Q_2\ldots Q_n},\tau_P \otimes \sigma_{Q_1}\otimes\sigma_{Q_2}\ldots \otimes \sigma_{Q_n}) \geq \frac{1}{1+\frac{2^k}{n}}.$$ In particular, for $\delta \in (0,1/3) $ and $ n= \lceil\frac{2^k}{\delta}\rceil$, 
 $$\relent{\tau_{PQ_1Q_2\ldots Q_n}}{\tau_P \otimes \sigma_{Q_1}\otimes\sigma_{Q_2}\ldots \otimes \sigma_{Q_n}} \leq \log(1+\delta)$$ and
$$ \F^{2}(\tau_{PQ_1Q_2\ldots Q_n},\tau_P \otimes \sigma_{Q_1}\otimes\sigma_{Q_2}\ldots \otimes \sigma_{Q_n}) \geq 1 - \delta.$$ 
\end{lemma}
\bigskip

The proof is as follows. 
\begin{proof}[Proof of Convex-split Lemma] 
We use the abbreviation $\sigma^{-j} \defeq \sigma_{Q_1}\ldots\otimes\sigma_{Q_{j-1}}\otimes \sigma_{Q_{j+1}}\ldots \otimes \sigma_{Q_n}$ and $\sigma\defeq \sigma_{Q_1}\otimes \sigma_{Q_2}\ldots \sigma_{Q_n}$. Then $\tau_{PQ_1Q_2\ldots Q_n}=\frac{1}{n}\sum_{j=1}^n\rho_{PQ_j}\otimes \sigma^{-j}$. Now, we use Lemma \ref{relentconcav} to express
\begin{eqnarray}
\label{eq:convsplit}
&& \relent{\tau_{PQ_1\ldots Q_n}}{\rho_P\otimes \sigma} = \frac{1}{n}\sum_j \relent{\rho_{PQ_j}\otimes \sigma^{-j}}{\rho_P\otimes\sigma}  - \frac{1}{n}\sum_j\relent{\rho_{PQ_j}\otimes \sigma^{-j}}{\tau_{PQ_1Q_2\ldots Q_n}}.
\end{eqnarray}
The first term in the summation on right hand side, $\relent{\rho_{PQ_j}\otimes \sigma^{-j}}{\rho_P\otimes \sigma}$, is equal to $\relent{\rho_{PQ_j}}{\rho_P\otimes \sigma_{Q_j}}$. 

The second term $\relent{\rho_{PQ_j}\otimes \sigma^{-j}}{\tau_{PQ_1Q_2\ldots Q_n}}$ is lower bounded by $\relent{\rho_{PQ_j}}{\tau_{PQ_j}}$, as relative entropy decreases under partial trace. But observe that $\tau_{PQ_j} = \frac{1}{n}\rho_{PQ_j}+(1-\frac{1}{n})\rho_P\otimes\sigma_{Q_j}$. By assumption, $\rho_{PQ_j} \leq 2^k\rho_P\otimes \sigma_{Q_j}$. Hence $\tau_{PQ_j} \leq (1+\frac{2^k-1}{n})\rho_P\otimes\sigma_{Q_j}$. Since $\log(A)\leq \log(B)$ if $A \leq B$ for positive semidefinite matrices $A$ and $B$ (see for example, \cite{carlen}), we have
\begin{eqnarray*}
&&\relent{\rho_{PQ_j}}{\tau_{PQ_j}} = \Tr(\rho_{PQ_j}\log\rho_{PQ_j}) - \Tr(\rho_{PQ_j}\log\tau_{PQ_j}) \\ &&\geq \Tr(\rho_{PQ_j}\log\rho_{PQ_j})  - \Tr(\rho_{PQ_j}\log(\rho_P\otimes\sigma_{Q_j})) - \log(1 + \frac{2^k-1}{n}) \\ && = \relent{\rho_{PQ_j}}{\rho_P\otimes\sigma_{Q_j}}- \log(1 + \frac{2^k-1}{n}) .
\end{eqnarray*}
Using in Equation \ref{eq:convsplit}, we find that
\begin{eqnarray*}
\relent{\tau_{PQ_1Q_2\ldots Q_n}}{\rho_P\otimes \sigma}&\leq& \frac{1}{n}\sum_j \relent{\rho_{PQ_j}}{\rho_P\otimes\sigma_{Q_j}} - \frac{1}{n}\sum_j\relent{\rho_{PQ_j}}{\rho_P\otimes\sigma_{Q_j}}+ \log(1 + \frac{2^k-1}{n})\\ &=& \log(1 + \frac{2^k-1}{n}).
\end{eqnarray*}
Thus, the lemma follows.  
\end{proof}

\textit{Connection to previous work:}
Following result appears as main theorem in the work of Csiszar \textit{et. al.}\cite{petz07},
 $$\lim_{n\rightarrow\infty}\relent{\tau_{Q_1Q_2\ldots Q_n}}{\sigma_{Q_1}\otimes\sigma_{Q_2}\ldots \sigma_{Q_n}}=0.$$ This is a special case of convex-split lemma in the limit $\delta\rightarrow 0$ (and hence $n\rightarrow \infty$) when the register $P$ is trivial. But it is also equivalent to convex-split lemma in the limit $\delta\rightarrow 0$ (and hence $n\rightarrow \infty$), as we argue below. Given an arbitrary hermitian operator $M\in \mathcal{L}(P)$, consider the normalized states $\rho'_{Q}=\frac{\Tr_P(M\rho_{PQ})}{\Tr(M\rho_P)}$ and $\tau'_{Q_1Q_2\ldots Q_n} = \frac{\Tr_P(M\tau_{PQ_1Q_2\ldots Q_n})}{\Tr(M\tau_P)}$. It is easy to observe that 
\begin{eqnarray*}
&& \tau'_{Q_1Q_2\ldots Q_n} \defeq \\ &&\frac{1}{n}\sum_{j=1}^n \rho'_{Q_j}\otimes\sigma_{Q_1}\otimes\ldots\otimes\sigma_{Q_{j-1}}\otimes\sigma_{Q_{j+1}}\otimes\ldots\otimes\sigma_{Q_n}
\end{eqnarray*}

 From the main theorem in \cite{petz07}, this state is arbitrarily close to $\sigma_{Q_1}\otimes\sigma_{Q_2}\ldots \otimes\sigma_{Q_n}$, for large enough $n$. This means that any measurement $M\in\mathcal{L}(P)$ on the state $\tau_{PQ_1Q_2\ldots Q_n}$  does not change the marginal on registers $Q_1Q_2\ldots Q_n$. Thus registers $P$ and $Q_1Q_2\ldots Q_n$ are independent in the state $\tau_{PQ_1Q_2\ldots Q_n}$. This coincides with the statement of convex-split lemma if we let $\delta\rightarrow 0$ (and hence $n\rightarrow \infty$).

\section{Compression of one-way quantum message}

Consider a state $\Phi_{RAMB}$ shared between $\Alice(AM)$, $\Bob(B)$ and $\Referee(R)$. The register $M$ serves as a message register, which $\Alice$ sends to $\Bob$. Following theorem shows that this message can be compressed. An idea of the proof appears in the Figure \ref{fig:convexexplanation}.

\begin{figure}[ht]
\centering
\begin{tikzpicture}[xscale=1.2,yscale=1.4]

\draw[fill] (1.0,8.5) circle [radius=0.05];
\draw[fill] (0.5,7.5) circle [radius=0.05];
\draw[fill] (-0.5,7.5) circle [radius=0.05];
\draw[thick] (1.0,8.5) -- (0.5,7.5);
\draw[thick] (1.0,8.5) -- (-0.5,7.5);
\draw[thick] (-0.5,7.5) -- (0.5,7.5);
\node at (1.2,7.6) {$\ket{\Phi}_{RBAM}$};

\draw[fill] (2.0,7) circle [radius=0.05];
\draw[fill] (0.0,6) circle [radius=0.05];
\draw[fill] (0.0,7) circle [radius=0.05];
\draw[fill] (2.0,6) circle [radius=0.05];
\draw[fill] (2.0,5) circle [radius=0.05];
\draw[fill] (0.0,5) circle [radius=0.05];
\draw[fill] (2.0,2) circle [radius=0.05];
\draw[fill] (0.0,2) circle [radius=0.05];
\draw[thick] (0,7) rectangle (2,2);

\draw[fill] (1.0,4.1) circle [radius=0.02];
\draw[fill] (1.0,3.7) circle [radius=0.02];
\draw[fill] (1.0,3.3) circle [radius=0.02];
\draw[fill] (1.0,2.9) circle [radius=0.02];

\node at (1.1,8.3) {RB};
\node at (-0.5,7.7) {A};
\node at (0.5,7.3) {M};
\node at (-0.3,7) {$L_1$};
\node at (-0.3,6) {$L_2$};
\node at (-0.3,5) {$L_3$};
\node at (-0.3,2) {$L_n$};
\node at (2.3,7) {$M_1$};
\node at (2.3,6) {$M_2$};
\node at (2.3,5) {$M_3$};
\node at (2.3,2) {$M_n$};

\node at (1,9) {$\Referee$};
\node at (-0.5,8.5) {$\Alice$};
\node at (2,8.5) {$\Bob$};

\node at (1,1.5) {$\ket{\theta}_{L_1L_2\ldots L_nM_1M_2\ldots M_n}$};


\node at (4.8,4.5) {$\frac{1}{\sqrt{n}}$};

\draw[fill] (7.5,7) circle [radius=0.05];
\draw[fill] (6.5,8) circle [radius=0.05];
\draw[fill] (5.5,7) circle [radius=0.05];
\draw[thick] (6.5,8) -- (7.5,7);
\draw[thick] (5.5,7) -- (7.5,7);
\draw[thick] (6.5,8) -- (5.5,7);
\node at (6.6,7.3) {$\ket{\Phi}_{RBL_1M_1}$};
\node at (7.5,6.5) {$\otimes$};
\node at (6.6,6.2) {$\ket{\sigma}_{L_2M_2}$};
\draw[fill] (7.5,6) circle [radius=0.05];
\draw[fill] (5.5,6) circle [radius=0.05];
\draw[thick] (5.5,6)--(7.5,6);
\node at (7.5,5.5) {$\otimes$};
\draw[fill] (7.5,5) circle [radius=0.05];
\draw[fill] (5.5,5) circle [radius=0.05];
\draw[thick] (5.5,5)--(7.5,5);
\node at (7.5,4.5) {$\otimes$};
\draw[fill] (7.5,4.1) circle [radius=0.02];
\draw[fill] (7.5,3.7) circle [radius=0.02];
\draw[fill] (7.5,3.3) circle [radius=0.02];
\draw[fill] (7.5,2.9) circle [radius=0.02];
\node at (7.5,2.5) {$\otimes$};
\draw[fill] (7.5,2) circle [radius=0.05];
\draw[fill] (5.5,2) circle [radius=0.05];
\draw[thick] (5.5,2)--(7.5,2);
\node at (6.6,2.2) {$\ket{\sigma}_{L_nM_n}$};
\node at (7.5,1.5) {$\otimes$};
\node at (7.5,1.0) {$\ket{1}_J$};

\node at (6.5,8.7) {$\Referee$};
\node at (5.5,8.2) {$\Alice$};
\node at (7.5,8.2) {$\Bob$};

\node at (6.5,8.2) {RB};
\node at (5.2,7) {$L_1$};
\node at (5.2,6) {$L_2$};
\node at (5.2,5) {$L_3$};
\node at (5.2,2) {$L_n$};
\node at (7.8,7) {$M_1$};
\node at (7.8,6) {$M_2$};
\node at (7.8,5) {$M_3$};
\node at (7.8,2) {$M_n$};

\node at (8.3,4.5) {$+$};

\node at (8.9,4.5) {$\frac{1}{\sqrt{n}}$};

\draw[fill] (11.5,7) circle [radius=0.05];
\draw[fill] (10.5,8) circle [radius=0.05];
\draw[fill] (9.5,6) circle [radius=0.05];
\draw[fill] (9.5,7) circle [radius=0.05];
\draw[thick] (11.5,7) -- (11.05,7);
\draw[thick] (10.95,7) -- (10.05,7);
\draw[thick] (9.95,7) -- (9.5,7);
\draw[thick] (10.5,8) -- (11.5,6);
\draw[thick] (11.5,6) -- (9.5,6);
\draw[thick] (9.5,6) -- (10.5,8);
\node at (10.5,6.2) {$\ket{\Phi}_{RBL_2M_2}$};
\node at (11.5,6.5) {$\otimes$};
\draw[fill] (11.5,6) circle [radius=0.05];
\node at (11.5,5.5) {$\otimes$};
\draw[fill] (11.5,5) circle [radius=0.05];
\draw[fill] (9.5,5) circle [radius=0.05];
\draw[thick] (9.5,5) -- (11.5,5);
\node at (11.5,4.5) {$\otimes$};
\draw[fill] (11.5,4.1) circle [radius=0.02];
\draw[fill] (11.5,3.7) circle [radius=0.02];
\draw[fill] (11.5,3.3) circle [radius=0.02];
\draw[fill] (11.5,2.9) circle [radius=0.02];
\node at (11.5,2.5) {$\otimes$};
\draw[fill] (11.5,2) circle [radius=0.05];
\draw[fill] (9.5,2) circle [radius=0.05];
\draw[thick] (9.5,2) -- (11.5,2);
\node at (10.5,2.2) {$\ket{\sigma}_{L_nM_n}$};
\node at (11.5,1.5) {$\otimes$};
\node at (11.5,1.0) {$\ket{2}_J$};

\node at (10.5,8.7) {$\Referee$};
\node at (9.5,8.2) {$\Alice$};
\node at (11.5,8.2) {$\Bob$};

\node at (10.5,8.2) {RB};
\node at (9.2,7) {$L_1$};
\node at (9.2,6) {$L_2$};
\node at (9.2,5) {$L_3$};
\node at (9.2,2) {$L_n$};
\node at (11.8,7) {$M_1$};
\node at (11.8,6) {$M_2$};
\node at (11.8,5) {$M_3$};
\node at (11.8,2) {$M_n$};

\draw[fill] (12.1,4.5) circle [radius=0.02];
\draw[fill] (12.5,4.5) circle [radius=0.02];
\draw[fill] (12.9,4.5) circle [radius=0.02];

\draw[->,thick] (2.9,4.5) -- (4.2,4.5);
\node at (3.5,4.7) {$V$};

\end{tikzpicture}
\caption{The state on left hand side $\ket{\Phi}_{RBAM} \otimes \ket{\theta}_{L_1L_2\ldots L_n M_1 \ldots M_n}$, a purification of $\Phi_{RB} \otimes\tau_{M_1 \ldots M_n}$. The state on right hand side is $\ket{\mu}_{JRBL_1L_2\ldots L_nM_1M_2\ldots M_n}$, a purification of $\tau_{RBM_1M_2\ldots M_n}$. Using convex-split lemma, Alice can apply an isometry $V$ on $\ket{\Phi}_{RBAM} \otimes \ket{\theta}_{L_1L_2\ldots L_n M_1 \ldots M_n}$ to obtain $\ket{\mu}_{JRBL_1L_2\ldots L_nM_1M_2\ldots M_n}$ with high fidelity.}
 \label{fig:convexexplanation}
\end{figure}
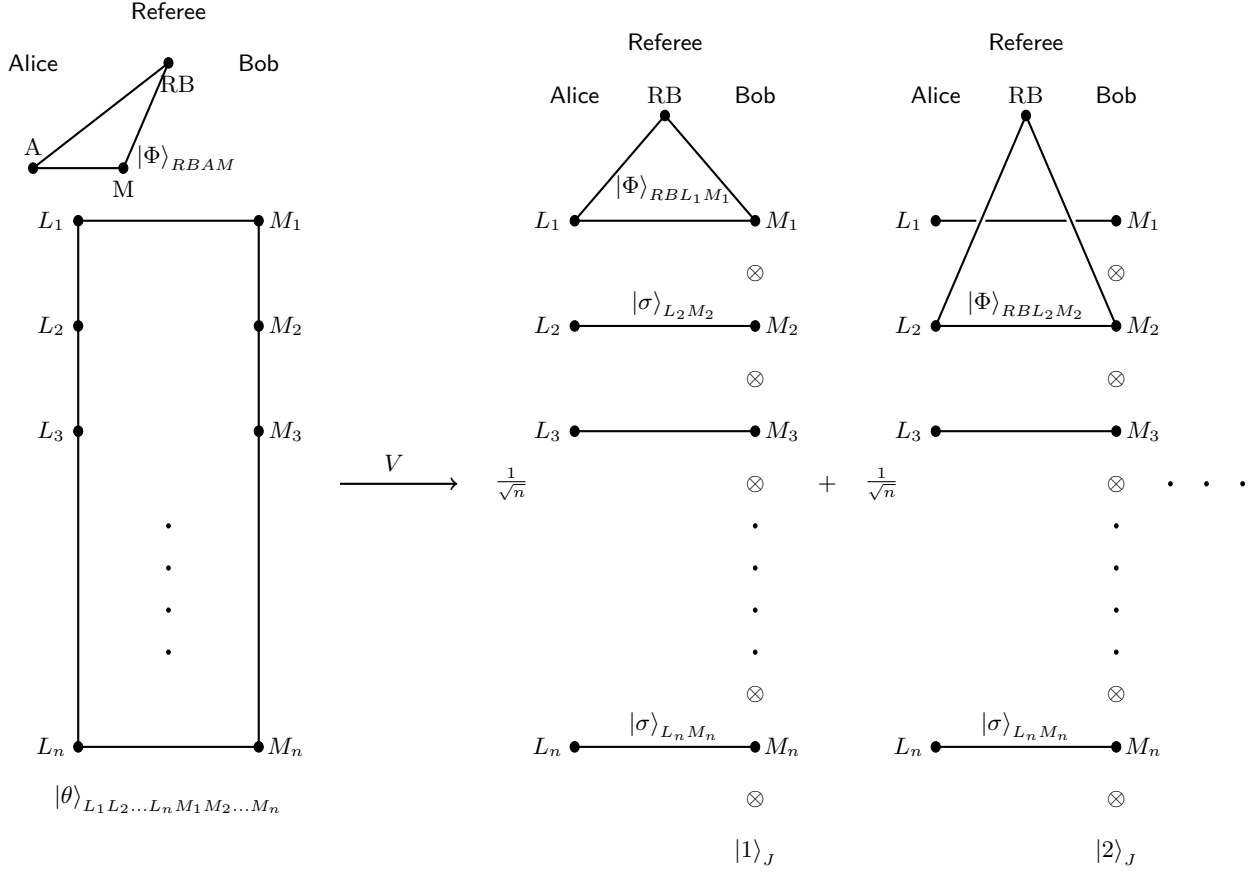

\begin{theorem}[Quantum message compression]\label{qmescompress}
There exists an entanglement-assisted one-way protocol $\P$, which takes as input $\ket{\Phi}_{RAMB}$ shared between three parties $\Referee~(R), \Bob~(B)$ and $\Alice~(AM)$ and outputs a state $\Phi'_{RAMB}$ shared between $\Referee~(R), \Bob~(BM)$ and $\Alice~(A)$ such that $\Phi'_{RAMB} \in  \ball{\eps}{\Psi_{RAMB}}$ and the number of qubits communicated by $\Alice$ to $\Bob$ in $\P$ is upper bounded by:
$$\frac{1}{2}\imax{RB}{M}_{\Phi} + \log\left(\frac{1}{\varepsilon}\right) .$$ 
\end{theorem}

\begin{proof}
Let $k \defeq \imax{RB}{M}_{\Phi}$, $\delta\defeq\varepsilon^2$ and $n \defeq \lceil\frac{2^k }{\delta}\rceil$. Let $\sigma_M$ be the state that achieves the infimum in the definition of $\imax{RB}{M}_{\Phi}$.  Consider the state, 
$$\mu_{RB M_1 \ldots M_n} \defeq \frac{1}{n}\sum_{j=1}^n \Phi_{RBM_j}\otimes\sigma_{M_1}\otimes\ldots\otimes\sigma_{M_{j-1}}\otimes\sigma_{M_{j+1}}\otimes\ldots \otimes\sigma_{M_n}.$$
Note that $\Phi_{RB} = \mu_{RB}$. Consider the following purification of $\mu_{RB M_1\ldots M_n}$, 
\begin{align*}
\lefteqn{\ket{\mu}_{RBJL_1\ldots L_n M_1\ldots M_n}}  \\ 
& & \defeq \frac{1}{\sqrt{n}}\sum_{j=1}^n\ket{j}_J\ket{\tilde{\Phi}}_{RBL_jM_j}\otimes \ket{\sigma}_{L_1M_1}\otimes \ldots\otimes\ket{\sigma}_{L_{j-1}M_{j-1}}\otimes\ket{\sigma}_{L_{j+1}M_{j+1}}\otimes\ldots\otimes\ket{\sigma}_{L_nM_n} 
\end{align*} 
Here, $\forall j\in[n]: \ket{\sigma}_{L_jM_j}$ is a purification of  $\sigma_{M_j}$ and $\ket{\tilde{\Phi}}_{RBL_jM_j}$ is a purification of $\Phi_{RBM_j}$. Consider the following protocol $\P_1$.
\begin{enumerate}
\item $\Alice$, $\Bob$ and $\Referee$ start by sharing the state $\ket{\mu}_{RBJL_1\ldots L_n M_1\ldots M_n}$ between themselves where $\Alice$ holds registers $JL_1\ldots L_n$, $\Referee$ holds the register $R$ and $\Bob$ holds the registers  $BM_1M_2\ldots M_n$.
\item $\Alice$ measures the register $J$ and sends the measurement outcome $j \in [n]$ to $\Bob$ using $\frac{\log(n)}{2}$ qubits  of quantum communication. $\Alice$ and $\Bob$ employ superdense coding (\cite{bennett92}) using fresh entanglement to achieve this. 
\item $\Alice$ swaps registers $L_j$ and $L_1$ and $\Bob$ swaps registers $M_j$ and $M_1$. Note that the joint state on the registers $RBL_1F_1$ at this stage is $\ket{\tilde{\Phi}}_{RBL_1M_1}$.
\item $\Alice$ applies an isometry $V: \H_{L_1} \rightarrow \H_{A}$ on the state $\ket{\tilde{\Phi}}_{RBL_1M_1}$ such that the joint state in registers $RAM_1B$ is $\Phi_{RBAM_1}$, as given by Uhlmann's theorem (Fact~\ref{uhlmann})
\end{enumerate}   

Consider the state,
$$\xi_{RB M_1\ldots M_n}  \defeq \Phi_{RB}\otimes \sigma_{M_1}\ldots \otimes \sigma_{M_n}. $$ 
Let $\ket{\theta}_{L_1\ldots L_n M_1\ldots M_n} = \ket{\sigma}_{L_1M_1}\otimes \ket{\sigma}_{L_2M_2}\ldots \ket{\sigma}_{L_nM_n}$ be a purification of  $\sigma_{M_1}\otimes \ldots \sigma_{M_n}$. Let 
 $$\ket{\xi}_{RABM L_1\ldots L_nM_1\ldots M_n} \defeq  \ket{\Phi}_{RABM} \otimes \ket{\theta}_{L_1\ldots L_n M_1\ldots M_n} .$$ 
Using convex-split lemma (Lemma ~\ref{convexcomb}) and choice of $n$ we have,
$$\F^2(\xi_{RBM_1\ldots M_n}, \mu_{RB M_1\ldots M_n})      \\ 
 \geq 1 - \varepsilon^2 .$$ 
Let $\ket{\xi'}_{RBJ L_1\ldots L_nM_1\ldots M_n}$ be a purification of $\xi_{RB M_1\ldots M_n}$ (guaranteed by Uhlmann's theorem, Fact~\ref{uhlmann}) such that,
$$\F^2(\ketbra{\xi'}_{RBJ L_1\ldots L_nM_1\ldots M_n},\ketbra{\mu}_{RBJ L_1\ldots L_nM_1\ldots M_n}) = \F^2(\xi_{RBM_1\ldots M_n}, \mu_{RB M_1\ldots M_n})  \geq 1- \varepsilon^2.$$ 
Let $V': \H_{AML_1\ldots L_n} \rightarrow \H_{J L_1\ldots L_n}$ be an isometry (guaranteed by Uhlmann's theorem, Fact~\ref{uhlmann})  such that,
$$V'\ket{\xi}_{RABM L_1\ldots L_nM_1\ldots M_n} = \ket{\xi'}_{RBJ L_1\ldots L_nM_1\ldots M_n} .$$ 
Consider the following protocol $\P$.
\begin{enumerate}
\item $\Alice$, $\Bob$ and $\Referee$ start by sharing the state $\ket{\xi}_{RABM L_1\ldots L_nM_1\ldots M_n} $ between themselves where $\Alice$ holds registers $AML_1\ldots L_n$, $\Referee$ holds the register $R$ and $\Bob$ holds  the registers  $BM_1\ldots M_n$. Note that  $\ket{\Psi}_{RABM}$ is provided as input to the protocol and $\ket{\theta}_{L_1\ldots L_n M_1\ldots M_n}$ is additional shared entanglement between $\Alice$ and $\Bob$. 
\item $\Alice$ applies isometry $V'$ to obtain state $\ket{\xi'}_{RBJL_1\ldots L_nM_1\ldots M_n}$, where $\Alice$ holds registers $JL_1\ldots L_n$, $\Referee$ holds the register $R$ and $\Bob$ holds  the registers  $BM_1\ldots M_n$.
\item $\Alice$ and $\Bob$ simulate protocol $\P_1$ from Step 2. onwards.
\end{enumerate}
Let $\Phi'_{RABM}$ be the output of protocol $\P$. Since quantum maps (the entire protocol $\P_1$ can be viewed as a quantum map from input to output) do not decrease fidelity (monotonicity of fidelity under quantum operation, Fact~\ref{fact:monotonequantumoperation}), we have,
\begin{equation}
\F^2(\Phi_{RABM}, \Phi'_{RABM}) \geq  \F^2(\ketbra{\xi'}_{RBJ L_1\ldots L_nM_1\ldots M_n},\ketbra{\mu}_{RBJ L_1\ldots L_nM_1\ldots M_n}) \geq 1- \varepsilon^2.
\label{eq:prot}
\end{equation}
This implies $\Phi_{RABM} \in \ball{\eps}{\ketbra{\Psi}_{RABC}}$.  

The number of qubits communicated by $\Alice$ to $\Bob$ in $\P$ is upper bounded by:
$$\frac{\log(n)}{2} \leq \frac{1}{2}\imax{RB}{M}_{\Phi} + \log\left(\frac{1}{\varepsilon}\right) .$$ 

\end{proof}

\section{Communication bounds on quantum state redistribution}

We begin with definition of quantum state redistribution. Please note that we allow $\Alice$ and $\Bob$ to share arbitrary prior entanglement. In comparison, the previous works \cite{Berta14,Oppenheim14} use EPR states and also take into account the amount of entanglement used by the protocol. 

\begin{definition}[Quantum state redistribution]
\label{def:qstateredistribution}
The quantum state $\ket{\Psi}_{RABC}\in\mathcal{D}(RABC)$ is shared between three parties $\Referee~(R), \Bob~(B)$ and $\Alice~(AC)$. In addition, $\Alice$ and $\Bob$ are allowed to share an arbitrary pure state $\ket{\theta}_{S^1_AS^1_B}$, where register $S^1_A$ belongs to $\Alice$ and register $S^1_B$ belongs to $\Bob$. Let $M$ represent the message register. $\Alice$ applies an encoding map $\E:\mathcal{L}(ACS^1_A)\rightarrow \mathcal{L}(AM)$ and sends the message $M$ to $\Bob$. $\Bob$ applies a decoding map $\mathcal{D}:\mathcal{L}(MBS^1_B)\rightarrow \mathcal{L}(BC)$. The resulting state $\Phi_{RABC}$ is the \textit{output} of the protocol. Quantum communication cost of the protocol is $\log|M|$.
\end{definition}

Using Stinespring representation (Fact \ref{stinespring}), the quantum maps $\E$ and $\mathcal{D}$ can be realized as unitary operations using additional ancillas. Let the ancillary register needed for map $\E$ by $\Alice$ be $S^2_A$, holding the state $\theta'_{S^2_A}$, and the ancillary register needed for map $\mathcal{D}$ by $\Bob$ be $S^2_B$, holding the state $\theta'_{S^2_B}$. Introduce registers $S_A\defeq S^1_AS^2_A$ and $S_B\defeq S^1_BS^2_B$. Let the joint state in registers $S_AS_B$ be $\ket{\theta}_{S_AS_B}$. Then following is equivalent to Definition \ref{def:qstateredistribution}. $\Alice$ applies a unitary $U_{ACS_A}$ on her registers, leading to the registers $AMT_A$ on her side (with $MT_A\equiv CS_A$). She sends $M$ to Bob, who applies a unitary $V_{MBS_B}$ and discards all his registers except $BC$. Let the registers discarded by $\Bob$ be $T_B$. The output of protocol is the state $\Phi_{RABC}$ in register $RABC$. Figure $3$ elaborates upon this description. 

\begin{figure}
\centering
\begin{tikzpicture}
\node at (0,6) {$\Psi_{RABC}$};
\draw[thick] (0.8,6) -- (2,10);
\draw[thick] (0.8,6) -- (2,7);
\draw[thick] (0.8,6) -- (2,5.5);
\draw[thick] (0.8,6) -- (2.2,0.5);
\draw[thick] (2.2,0.5) -- (7.5,0.5);
\draw[thick] (2,5.5) -- (3.8,5.5);
\draw[thick] (2,7) -- (3.8,7);
\draw[thick] (2,10) -- (10,10);

\node at (2,4) {$\theta_{S^1_AS^1_B}$};
\draw[thick] (2.5,4) -- (2.8,4.8);
\draw[thick] (2.5,4) -- (3,1.8);
\draw[thick] (2.8,4.8) -- (3.8,4.8);
\draw[thick] (3,1.8) -- (7.5,1.8);

\draw[fill=green, ultra thick] (3.8,8) rectangle (5.2,4.2);
\node at (2.7,7.6) {$\theta'_{S^2_A}$};
\draw[thick] (3,7.6) -- (3.8,7.6);
\draw[thick] (5.2,4.4) -- (6,4.4);
\draw[thick] (6,4.4) -- (6.5,3);
\draw[thick] (6.5,3) -- (7.5,3);
\draw[thick] (5.2,5.5) -- (10,5.5);
\draw[thick] (10,5.5) -- (10.8,4);
\draw[thick] (5.2,7) -- (10,7);

\draw[fill=green, ultra thick] (7.5,3.8) rectangle (9,0);
\node at (6.2,1) {$\theta'_{S^2_B}$};
\draw[thick] (6.5,1) -- (7.5,1);
\draw[thick] (9,0.5) -- (10,0.5);
\draw[thick] (9,1.5) -- (10,1.5);
\draw[thick] (9,3) -- (10,3);
\draw[thick] (10,3) -- (10.8,4);

\node at (13.7,5) {$\Phi_{RABC}$};
\draw[thick] (10,7) -- (13,5);
\draw[thick] (10,0.5) -- (13,5);
\draw[thick] (10,1.5) -- (13,5);
\draw[thick] (10,10) -- (13,5);

\node at (1,8) {$\text{R}$};
\node at (1.6,6.3) {$\text{A}$};
\node at (1.6,5.4) {$\text{C}$};
\node at (1.2,3) {$\text{B}$};
\node at (3.2,4.5) {$\text{S}^1_A$};
\node at (3.2,3) {$\text{S}^1_B$};
\node at (3.3,8) {$\text{S}^2_A$};
\node at (7,1.3) {$\text{S}^2_B$};
\node at (6.6,3.6) {$\text{M}$};
\node at (3.2,3) {$\text{S}^1_B$};
\node at (6,5.8) {$\text{T}_A$};
\node at (6,7.2) {$\text{A}$};
\node at (9.5,0.7) {$\text{B}$};
\node at (9.5,1.7) {$\text{C}$};
\node at (9.5,3.2) {$\text{T}_B$};
\node at (4.5,6) {$\text{U}$};
\node at (8.3,2) {$\text{V}$};

\draw [gray] (-1,8.5) -- (14,8.5);
\draw [gray] (-1,4) -- (1,4);
\draw [gray] (2.7,4) -- (14,4);
\node at (13.5,9) {$\Referee$};
\node at (13.5,7) {$\Alice$};
\node at (13.5,3) {$\Bob$};
\end{tikzpicture}
\caption{Graphical representation of one-way entanglement assisted quantum state redistribution.}
\label{fig:stateredistpic}
\end{figure}
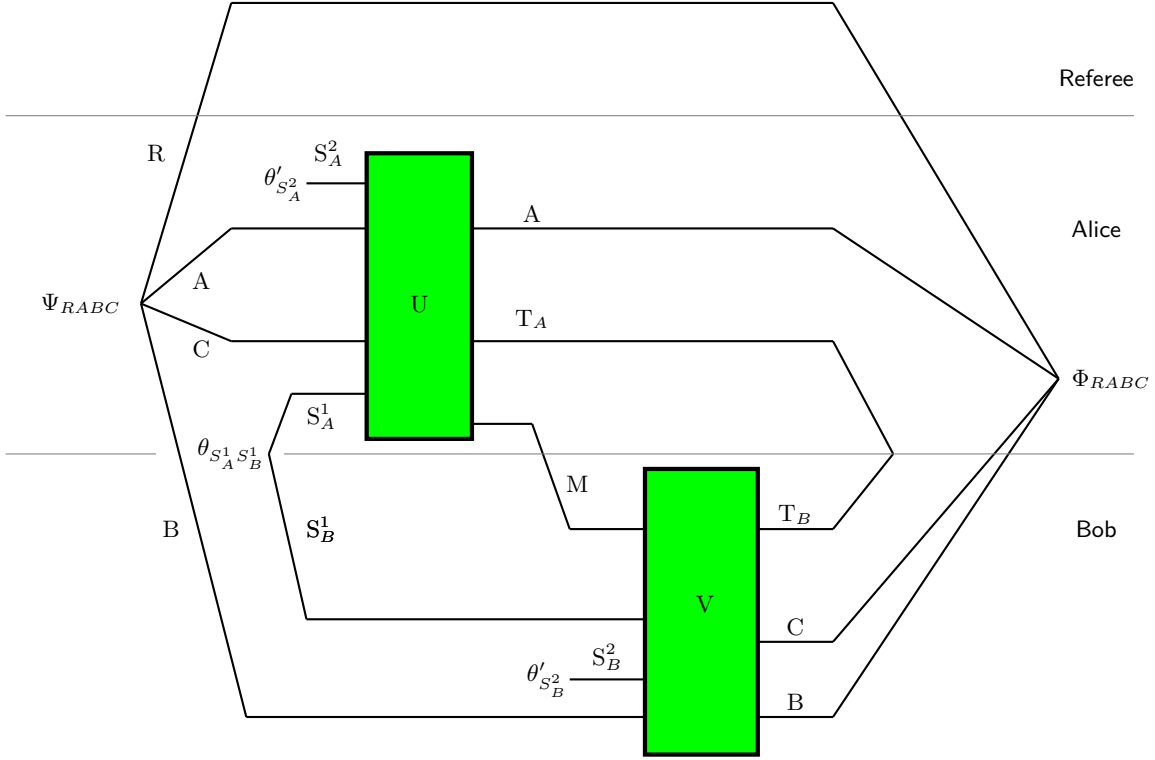

Before proceeding to our upper and lower bounds, we present the following definition. 
\begin{definition} 
\label{redistbound}
Let $\varepsilon \geq 0$ and $\Psi_{RABC}\in\mathcal{D}(RABC)$ be a pure state. Define,
\begin{eqnarray*}  
\mathrm{Q}^{\varepsilon}_{\ket{\Psi}_{RABC}} &\defeq& \inf_{T, U_{BCT},\sigma'_T,\kappa_{RBCT}} \imax{RB}{CT}_{\kappa_{RBCT}} \\
& =& \inf_{T, U_{BCT},\sigma'_T,\sigma_{CT},\kappa_{RBCT}} \dmax{\kappa_{RBCT}}{\kappa_{RB}\otimes\sigma_{CT}} 
\end{eqnarray*}
with the conditions $U_{BCT}\in \mathcal{U}(BCT),  \sigma'_T \in \mathcal{D}(T), \sigma_{CT} \in \mathcal{D}(CT)$ and
$$(I_R \otimes U_{BCT}) \kappa_{RBCT}(I_R \otimes U^{\dagger}_{BCT}) \in \ball{\eps}{\Psi_{RBC}\otimes\sigma'_T}, \kappa_{RB}=\Psi_{RB}.$$ 
\end{definition}

\subsection*{Lower bound}
\label{sec:lower}
We have the following lower bound result.
\begin{theorem}[Lower bound]
\label{lowerbound}
Let $\varepsilon > 0$ and $\Psi_{RABC}\in\mathcal{D}(RABC)$ be a pure state.  Let $\Q$ be an entanglement-assisted one-way protocol (with communication from $\Alice$ to $\Bob$), which takes as input $\ket{\Psi}_{RABC}$ shared between three parties $\Referee~(R), \Bob~(B)$ and $\Alice~(AC)$ and outputs a state $\Phi_{RABC}$ shared between $\Referee~(R), \Bob~(BC)$ and $\Alice~(A)$ such that $\Phi_{RABC} \in  \ball{\eps}{\Psi_{RABC}}$. The number of qubits communicated by $\Alice$ to $\Bob$ in $\Q$ is lower bounded by:
$$\frac{1}{2}\mathrm{Q}^{\varepsilon}_{\ket{\Psi}_{RABC}}.$$ 
\end{theorem}
\begin{proof}
Protocol $\Q$ can be written as follows (see Figure \ref{fig:stateredistpic} ):
\begin{enumerate}
\item $\Alice$ and $\Bob$ get as input $\ket{\Psi}_{RABC}$ shared between $\Alice$ (AC), $\Referee$ (R) and $\Bob$ (B). In addition $\Alice$ and $\Bob$ use shared entanglement and local ancillas for their protocol. Let these additional resources be represented by a pure state $\ket{\theta}_{S_AS_B}$ where register $S_A$ is held by $\Alice$ and register $S_B$ is held by $\Bob$. 
\item  $\Alice$ applies a unitary $U_{ACS_A}$ on the registers $ACS_A$.  Let $\kappa_{RMAT_ABS_B}$ be the joint state at this stage shared between $\Alice$ ($MAT_A$), $\Referee$ ($R$) and $\Bob$ ($BS_B$),  where  $MT_A  \equiv CS_A$. Note that $\kappa_{RB}=\Psi_{RB}$ and $\kappa_{RBS_B}=\Psi_{RB}\otimes \theta_{S_B}$. 
\item $\Alice$ sends the  message register $M$ to $\Bob$. 
\item $\Bob$ applies a unitary $V_{BS_BM}$ on the registers $BS_BM$. Let $\Phi_{RABCT_AT_B}$ be the joint state at this stage shared between $\Alice$ ($AT_A$), $\Referee$ (R) and $\Bob$ ($BCT_B$) where $S_BM \equiv CT_B$.
\item The state $\Phi_{RABC}$ is considered the output of the protocol $\Q$.
\end{enumerate}
Using Fact~\ref{dmaxuppbound}, we know that there exists a state $\omega_M$, such that:
\begin{equation} \label{eq:low-1}
2 \log |M|  \geq \dmax{\kappa_{RBS_BM}}{\kappa_{RBS_B}\otimes\omega_M}  =  \dmax{\kappa_{RBS_BM}}{\Psi_{RB} \otimes \theta_{S_B}\otimes\omega_M} .
\end{equation}
We have $\F^2(\Phi_{RABC},\ketbra{\Psi}_{RABC}) \geq 1 - \varepsilon^2 $ and $\ketbra{\Psi}_{RABC}$ is a pure state. From Lemma~\ref{lem:pureclose} and monotonicity of fidelity under quantum operation (Fact~\ref{fact:monotonequantumoperation}) we get a  state $\sigma'_{T_B}$ such that,
 $$\F^2(\Phi_{RBCT_B}, \Psi_{RBC}\otimes\sigma'_{T_B}) \geq 1 -  \varepsilon^2 .$$ 
We have,
\begin{equation} \label{eq:low-2}
\Phi_{RBCT_B} =  (I_R \otimes V_{BS_BM}) \kappa_{RBS_BM} (I_R \otimes V^{\dagger}_{BS_BM}),~\kappa_{RB} = \Psi_{RB} .
\end{equation}
Recall that $S_BM \equiv CT_B$.  Define $\sigma_{CT_B} \defeq \theta_{S_B}\otimes\omega_M$. Eq.~\eqref{eq:low-1} and Eq.~\eqref{eq:low-2} imply, 
$$ 2 \log |M|  \geq  \dmax{\kappa_{RBCT_B}}{\Psi_{RB} \otimes \sigma_{CT_B}} ,$$ 
with the conditions
\begin{equation*} 
\F^2(\Phi_{RBCT_B},\Psi_{RBC}\otimes\sigma'_{T_B})>1-\varepsilon^2, \Phi_{RBCT_B}= (I_R \otimes V_{BCT_B}) \kappa_{RBCT_B}(I_R \otimes V^{\dagger}_{BCT_B}),\kappa_{RB} = \Psi_{RB} .
\end{equation*}
From above and the definition of $Q^{ \varepsilon}_{\ket{\Psi}_{RABC}}$ we conclude
$$ \log |M|\geq \frac{1}{2}Q^{ \varepsilon}_{\ket{\Psi}_{RABC}} .$$ 
\end{proof}

\subsection*{Upper bound}
\label{sec:upper}

We show a nearly matching upper bound on the quantum communication cost of quantum state redistribution.

\begin{theorem}[Upper bound] \label{thm:main}
Let $\varepsilon \in (0,1/3)$ and $\Psi_{RABC}\in\mathcal{D}(RABC)$ be a pure state. There exists an entanglement-assisted one-way protocol $\P$, which takes as input $\ket{\Psi}_{RABC}$ shared between three parties $\Referee~(R), \Bob~(B)$ and $\Alice~(AC)$ and outputs a state $\Phi_{RABC}$ shared between $\Referee~(R), \Bob~(BC)$ and $\Alice~(A)$ such that $\Phi_{RABC} \in  \ball{2\eps}{\Psi_{RABC}}$. The number of qubits communicated by $\Alice$ to $\Bob$ in $\P$ is upper bounded by:
$$\frac{1}{2}\mathrm{Q}^{\varepsilon}_{\ket{\Psi}_{RABC}} + \log\left(\frac{2}{\varepsilon}\right) .$$ 
\end{theorem}

\begin{proof}
The definition of $\mathrm{Q}^{\varepsilon}_{\ket{\Psi}_{RABC}}$ involves an infimum over various quantities. There exists a collection $(T, U_{BCT}, \sigma'_T,\sigma_{CT}, \kappa_{RBCT})$ along with the conditions,
$$(I_R \otimes U_{BCT}) \kappa_{RBCT}(I_R \otimes U^{\dagger}_{BCT}) \in \ball{\eps}{\Psi_{RBC}\otimes\sigma'_T}, \kappa_{RB} = \Psi_{RB},$$ such that $\imax{RB}{CT}_{\kappa} \leq \Q^{\eps}_{\Psi_{RABC}} + 1$. 

Define the state $$\rho_{RBCT} \defeq (I_R \otimes U_{BCT}) \kappa_{RBCT}(I_R \otimes U^{\dagger}_{BCT}).$$ 

Since $\kappa_{RB} = \Psi_{RB}$, then for any purification $\ket{\kappa}_{RBCTS}$ of $\kappa_{RBCT}$, there exists an isometry $V_1: \H_{AC}\rightarrow \H_{CTS}$ such that 
\begin{equation}
\label{eq:nearbygood}
\ketbra{\kappa}_{RBCTS} = V_1\Psi_{RBAC}V_1^{\dagger}
\end{equation}
We start with the following protocol $\P_1$.
\begin{enumerate}
\item $\Alice (CTS)$, $\Bob(B)$ and $\Referee(R)$ start with the state $\ket{\kappa}_{RBCTS}$ and shared entanglement as required in the protocol described in Theorem \ref{qmescompress}.
\item Using the protocol described in Theorem \ref{qmescompress}, the parties produce a state $\kappa'_{RBCTS}$ with registers $BCT$ belonging to $\Bob$, $S$ belonging to $\Alice$ and $R$ belonging to $\Referee$, such that $\F^2(\kappa'_{RBCTS}, \kappa_{RBCTS})\geq 1-\eps^2$. In other words,
\begin{equation}
\label{eq:kappafid}
\P(\kappa'_{RBCTS}, \kappa_{RBCTS}) \leq \eps
\end{equation}

\item Bob applies the unitary $U_{BCT}$ on registers $BCT$.   
\end{enumerate}

The number of qubits communicated in $\P_1$ is $\frac{1}{2}\imax{RB}{CT}_{\kappa} + \log(\frac{1}{\eps})$. 

At the end of the protocol, the state in registers $RBCT$ is $U_{BCT}\kappa'_{RBCT}U^{\dagger}_{BCT}$. By definition of $\rho_{RBCT}$, the relation $\P(\rho_{RBCT}, \Psi_{RBC}\otimes \sigma'_T) \leq \eps$ and Equation \ref{eq:kappafid}, we find (using triangle inequality for purified distance (Fact~\ref{fact:trianglepurified})) that 
$$\P(\Psi_{RBC}\otimes \sigma'_T, U_{BCT}\kappa'_{RBCT}U^{\dagger}_{BCT})\leq 2\eps.$$

Thus, there exists an isometry $V_2: \H_{S}\rightarrow \H_{AE}$ such that for a purification $\ket{\sigma'}_{ET}$ of $\sigma_T$, 
\begin{equation}
\label{eq:finalfid}
\P(\Psi_{RABC}\otimes \ketbra{\sigma'}_{ET}, V_2\otimes U_{BCT}\kappa'_{RBCTS}U^{\dagger}_{BCT}\otimes V_2^{\dagger})\leq 2\eps
 \end{equation}

Now, we describe the protocol $\P$ that achieves the desired task.
\begin{enumerate}
\item $\Alice (AC)$, $\Bob(B)$ and $\Referee(R)$ start with the state $\ket{\Psi}_{RABC}$ and the shared entanglement as required
to run the protocol $\P_1$ below.
\item $\Alice$ applies the isometry $V_1$ on her registers. The parties run the protocol $\P_1$. Finally, $\Alice$ applies the isometry $V_2$ on her registers. 
\end{enumerate}

Let the final state produced in registers $RABC$ be $\Phi_{RABC}$. Using equations \ref{eq:finalfid} and \ref{eq:nearbygood}, we find that
$\P(\Psi_{RABC},\Phi_{RABC})\leq 2\eps$.   

Since the quantum communication cost of $\P$ is equal to the quantum communication cost of $\P_1$, the number of qubits communicated by $\Alice$ to $\Bob$ in $\P$ is upper bounded by:
$$\frac{\log(n)}{2} \leq \frac{1}{2}\mathrm{Q}^{\varepsilon}_{\ket{\Psi}_{RABC}} + \frac{1}{2}+ \log\left(\frac{1}{\varepsilon}\right) \leq \frac{1}{2}\mathrm{Q}^{\varepsilon}_{\ket{\Psi}_{RABC}} + \log\left(\frac{2}{\varepsilon}\right).$$ 

\end{proof}

\section{Communication bounds on quantum state splitting and quantum state merging}
\label{sec:splitmerge}
In this section, we describe near optimal bound for quantum communication cost of quantum state splitting and quantum state merging protocols. We recall that quantum state splitting is a special case of quantum state redistribution in which the register $B$ is trivial and quantum state merging is a special case of quantum state redistribution in which register $A$ is trivial.
\subsection*{Quantum state splitting}
We show the following lemma, which along with our upper bound (Theorem~\ref{thm:main}) and lower bound (Theorem~\ref{lowerbound}) immediately gives the desired upper and lower bound on quantum communication cost of quantum state splitting.
\begin{lemma} \label{lem:splitting} Let $\Psi_{RABC}\in \mathcal{D}(RABC)$ be a pure quantum state and let $B$ be a trivial register, that is, $|B|=1$. Then $\mathrm{Q}^{\varepsilon}_{\ket{\Psi}_{RABC}} = \imaxeps{R}{C}_{\Psi_{RC}}$.
\end{lemma}
\begin{proof} Since register $B$ is trivial, we drop the notation $B$ from the quantum states discussed below. Given the quantum state $\kappa_{RCT}$ as appearing in definition of $Q^{\eps}_{\ket{\Psi}_{RAC}}$ (Definition \ref{redistbound}), we define the state
$$ \rho_{RCT} \defeq (I_R \otimes U_{CT}) \kappa_{RCT} (I_R \otimes U^{\dagger}_{CT}).$$ 
It holds that $\rho_{RCT}\in \ball{\eps}{\Psi_{RC}\otimes\sigma'_T}$. Note that the condition $\kappa_{R}\in \ball{\eps}{\Psi_{R}}$ is now redundant (is implied by above using $\rho_R = \kappa_R$ and monotonicity of fidelity under quantum operation, Fact~\ref{fact:monotonequantumoperation}). Consider,
\begin{eqnarray*}
\mathrm{Q}^{\varepsilon}_{\ket{\Psi}_{RAC}} &=& \inf_{T, U_{CT},\sigma_{CT},\sigma'_T,\kappa_{RCT}} \dmax{\kappa_{RCT}}{\kappa_R\otimes \sigma_{CT}} \\ &=& \inf_{T, U_{CT},\sigma_{CT},\sigma'_T,\kappa_{RCT}} \dmax{(I_R\otimes U^{\dagger}_{CT})\rho_{RCT}(I_R\otimes U_{CT})}{\kappa_R\otimes \sigma_{CT}} \\ &=& \inf_{T, U_{CT},\sigma_{CT},\sigma'_T,\kappa_{RCT}} \dmax{\rho_{RCT}}{\kappa_R\otimes U_{CT}\sigma_{CT}U_{CT}^{\dagger}} \\ &=& \inf_{T, \mu_{CT},\sigma'_T,\kappa_{RCT}} \dmax{\rho_{RCT}}{\kappa_R\otimes \mu_{CT}} \quad (\text{with } \mu_{CT}\defeq U_{CT}\sigma_{CT}U_{CT}^{\dagger})\\
&=& \inf_{T,\sigma'_T, \rho_{RCT}\in \ball{\eps}{\Psi_{RC}\otimes\sigma'_T}} \imax{R}{CT}_{\rho_{RCT}} \quad (\text{using } \rho_R = \kappa_R)\\
&=& \inf_{T,\sigma'_T} \imaxeps{R}{CT}_{\Psi_{RC}\otimes \sigma'_T} .
\end{eqnarray*}
Now,
\begin{eqnarray*}
\imaxeps{R}{C}_{\Psi_{RC}} &\geq& \inf_{T,\sigma'_T} \imaxeps{R}{CT}_{\Psi_{RC}\otimes \sigma'_T} \quad (\mbox{by setting $T$ to be trivial register}) \\
&=&\inf_{T,\sigma'_T, \rho_{RCT}\in \ball{\eps}{\Psi_{RC}\otimes\sigma'_T}} \imax{R}{CT}_{\rho_{RCT}}  \\
&\geq& \inf_{\rho_{RC}\in \ball{\eps}{\Psi_{RC}}} \imax{R}{C}_{\rho_{RC}} \\&& (\text{using monotonicity of max-information under quantum operation, Fact}~\ref{imaxmonotone}) \\
&=& \imaxeps{R}{C}_{\Psi_{RC}}.
\end{eqnarray*} 
Therefore,
$$ \mathrm{Q}^{\varepsilon}_{\ket{\Psi}_{RAC}} =  \inf_{T,\sigma'_T} \imaxeps{R}{CT}_{\Psi_{RC}\otimes \sigma'_T} = \imaxeps{R}{C}_{\Psi_{RC}}  .$$
\end{proof}

\subsection*{Quantum state merging}
Now, we consider the case of quantum state merging. It has been noted in \cite{Renner11} that quantum state merging can be viewed as `time reversed' version of quantum state splitting, and their optimal quantum communication cost is the same. 
\begin{lemma}[\cite{Renner11}] \label{lem:splitsameasmerge}
Let $\varepsilon > 0$ be error parameter. Following two statements are equivalent, with registers $A$ and $B$ such that $A\equiv B$.
\begin{enumerate}
\item There exists an entanglement assisted quantum state splitting protocol $\P$ with quantum communication cost $c$, that starts with a state $\Psi_{RAC} \in\mathcal{D}(RAC)$, with $AC$ on $\Alice$'s side and $R$ on $\Referee$'s side, and outputs a state $\Phi_{RAC}$, with $C$ on $\Bob$'s side, such that $\Phi_{RAC}\in \ball{\eps}{\Psi_{RAC}}$.

\item There exists an entanglement assisted quantum state merging protocol $\Q$ with quantum communication cost $c$, that starts with the state $\Psi_{RBC} \in \mathcal{D}(RBC)$, with $C$ on $\Alice$'s side and $B$ on $\Bob$'s side, and outputs  a state $\Phi'_{RBC}$, with $(BC)$ on $\Bob$'s side, such that $\Phi'_{RBC}\in \ball{\eps}{\Psi_{RBC}}$. 
\end{enumerate}
\end{lemma}

\begin{proof}
We show that $(1)\implies (2)$. Let the protocol $\P$ start with the overall pure state $\Psi_{RAC}\otimes\mu_{E}$, where the register $E$ include shared entanglement and other ancilla registers used by $\P$. Let the final pure state of the protocol be $\Phi_{RACE}$, with $\F^2(\Phi_{RAC},\Psi_{RAC})\geq 1-\varepsilon^2$. To describe the quantum state merging protocol, we now relabel register $A$ with register $B$. Since  protocol $\P$ is a collection of unitary operations (which are invertible, see discussion after Definition \ref{def:qstateredistribution}),  it implies that there exists a protocol $\P'$ (which is inverse of the protocol $\P$) that starts with the state $\Phi_{RBCE}$, and leads to the state $\Psi_{RBC}\otimes \mu_{E}$ with $\F^2(\Psi_{RBC},\Phi_{RBC})\geq 1-\varepsilon^2$. From Uhlmann's theorem (Fact~\ref{uhlmann}), there exists a pure state $\mu'_{E}$ that satisfies $$\F^2(\Psi_{RBC}\otimes\mu'_{E},\Phi_{RBCE}) = \F^2(\Psi_{RBC},\Phi_{RBC})\geq 1-\varepsilon^2.$$ 
Let $\Q$ be a protocol that starts with the pure state $\Psi_{RBC}\otimes\mu'_{E}$, and then follows the protocol $\P'$. Let the overall state at the end of $\Q$ be $\Phi'_{RBCE}$.  Then,
$$\F^2(\Psi_{RBC},\Phi'_{RBC})\geq \F^2(\Psi_{RBC}\otimes\mu_{E},\Phi'_{RBCE})=\F^2(\Phi_{RBCE}, \Psi_{RBC}\otimes \mu'_{E})) \geq 1-\varepsilon^2.$$ 
It is clear that the communication between $\Alice$ and $\Bob$ is the same in $\P$ and $\Q$.

$(2)\implies (1)$ can be proved using similar arguments. 
\end{proof}

\section{Port-based teleportation}
We consider the problem of port-based teleportation, when the sender $\Alice$ and the receiver $\Bob$ know that the set of possible states to be teleported belong to the ensemble $\{p_i,\ketbra{\psi}^i\}_i$, with $\sum_{i} p_i=1$. $\Alice$ is given the state $\ketbra{\psi^i}$ with probability $p_i$ which she wishes to teleport to $\Bob$. 

Before proving our result, we will prove the following useful Lemma. It can be seen as a one-sided analogue of the relation between optimal fidelity of teleportation and maximal singlet fraction as proven in \cite{horodecki99}. 
\begin{lemma}
\label{lem:entfidaveragefid}
Given a quantum channel $\E:M\rightarrow M$ with Kraus-representation $\E(\rho) = \sum_k A_k \rho A_k^{\dagger}$ and an ensemble $\{p_i,\ketbra{\psi}_M^i\}_i$ with $\psi^i_M\in \mathcal{D}(M)$, define the state 
$\ket{\Psi}_{RM} \defeq \sum_i \sqrt{p_i}\ket{i}_R\ket{\psi}^i_M$. Then it holds that 
$$\bra{\Psi}_{RM} \E(\Psi_{RM})\ket{\Psi}_{RM} \leq \sum_i p_i \bra{\psi}^i_M\E(\psi^i_M)\ket{\psi}^i_M.$$ 
\end{lemma}

\begin{proof}
We proceed as follows.
\begin{eqnarray*}
\bra{\Psi}_{RM} \E(\Psi_{RM})\ket{\Psi}_{RM} & = & \sum_k |\bra{\Psi}_{RM}\mathrm{I}_R\otimes A_k\ket{\Psi}_{RM}|^2 
 =  \sum_k | \sum_i p_i \Tr(\psi^i_M A_k) |^2 \\ & \leq & \sum_k (\sum_i p_i)\cdot (\sum_i p_i |\Tr(\psi^i_M A_k)|^2) = \sum_i p_i \sum_k |\Tr(\psi^i_M A_k)|^2
\end{eqnarray*}

The inequality above is due to the Cauchy-Schwartz inequality. Now, we observe that 
$$\sum_i p_i \bra{\psi}^i_M\E(\psi^i_M)\ket{\psi}^i_M =  \sum_i p_i \sum_k |\Tr(\psi^i_M A_k)|^2,$$
which completes the proof.
\end{proof}

Now we proceed to our main theorem of this section.

\begin{theorem}
\label{portbased}
Consider an ensemble of pure quantum states $\{p_i,\ketbra{\psi}^i_M\}_{i}$, with $\psi^i_M\in \mathcal{D}(M)$. Introduce a register $R$ and define the state $\ket{\Psi}_{RM} \defeq \sum_i \sqrt{p_i}\ket{i}_R\ket{\psi}^i_M$.  Let $\sigma_M$ be an arbitrary state and $k \defeq \dmax{\Psi_{RM}}{\Psi_R\otimes \sigma_M}$. Suppose $\Alice$ and $\Bob$ share $n$ copies of a purification of $\sigma_M$. Then there exists a port-based teleportation protocol such that $\Bob$ outputs the register $M'\equiv M$ and for each $i$, the final state with $\Bob$ is $\phi^i_{M'}$ such that $\sum_i p_i \F^2(\psi^i_{M'},\phi^i_{M'})\geq 1-\frac{2^k}{n}$.  
\end{theorem}

\begin{proof}

We define the state $$\tau_{RM_1M_2\ldots M_n} \defeq \frac{1}{n}\sum_j \Psi_{RM_j}\otimes\sigma_{M_1}\otimes\ldots \sigma_{M_{j-1}}\otimes \sigma_{M_{j+1}}\ldots\otimes \sigma_{M_n}.$$ 

Consider the following purification of $\tau^i_{RM_1M_2\ldots M_n}$, $$\ket{\tau^i}_{JL_1L_2\ldots L_n R M_1M_2\ldots M_n}\defeq \frac{1}{\sqrt{n}}\sum_j \ket{j}_J\ket{\Psi}_{RM_j}\ket{\sigma}_{L_1M_1}\otimes\ldots\ket{\sigma}_{L_{j-1}M_{j-1}}\otimes \ket{0}_{L_j}\otimes\ket{\sigma}_{L_{j+1}M_{j+1}}\ldots\otimes\ket{\sigma}_{L_nM_n},$$
where $\ket{\sigma}_{L_iM_i}$ is a purification of $\sigma_{M_i}$ and $\ket{0}_{L_j}$ is some fixed state.  

From convex split lemma \ref{convexcomb}, it holds that 
$$\F^2(\tau_{RM_1M_2\ldots M_n}, \Psi_R\otimes\sigma_{M_1}\otimes \sigma_{M_2}\ldots \otimes \sigma_{M_n})\geq \frac{1}{1+\frac{2^{k}}{n}}.$$

Thus, there exists an isometry $V: \H_{ML_1L_2\ldots L_n}\rightarrow \H_{JL_1L_2\ldots L_n}$ (guaranteed by Uhlmann's theorem, Fact~\ref{uhlmann}), such that 
\begin{equation}
\label{eq:portbased}
\F^2(\ketbra{\tau}_{JL_1L_2\ldots L_n RM_1M_2\ldots M_n}, V\ketbra{\Psi}_{RM}\otimes\ketbra{\sigma}_{L_1M_1}\otimes \ketbra{\sigma}_{L_2M_2}\ldots \otimes \ketbra{\sigma}_{L_nM_n}V^{\dagger})\geq \frac{1}{1+\frac{2^{k}}{n}}.
\end{equation}

We consider the following protocol $\P$:
\begin{enumerate}
\item $\Alice$ and $\Bob$ share $n$ copies of the state $\ket{\sigma}_{LM}$ in registers $L_1M_1,L_2M_2,\ldots L_nM_n$. 
\item $\Alice$ applies the isometry $V$ and measures the register $J$. Then she sends the outcome $j$ to $\Bob$.
\item Upon receiving the outcome $j$, $\Bob$ picks up the register $M_j$ and swaps it with his output register $M'$. 
\end{enumerate}

Consider the action of $\P$ when the input to it is the state $\Psi_{RM}$. Let the state in the registers $RM'$ upon the completion of $\P$ be $\P(\Psi_{RM})$. From Equation \ref{eq:portbased} and monotonicity of fidelity under quantum map (Fact \ref{fact:monotonequantumoperation}), it holds that $\F^2(\P(\Psi_{RM}),\Psi_{RM})\geq \frac{1}{1+\frac{2^{k}}{n}} \geq 1-\frac{2^k}{n}$.

Since $\P$ is a quantum map, we can apply Lemma \ref{lem:entfidaveragefid} to conclude that 
$$\sum_i p_i \F^2(\phi^i_{M'},\psi^i_{M'})=\sum_i p_i \F^2(\P(\psi^i_M),\psi^i_M) \geq \F^2(\P(\Psi_{RM}),\Psi_{RM}) \geq 1-\frac{2^k}{n}.$$

This proves the theorem.

\end{proof}

\bibliographystyle{alpha}
\bibliography{state-redistribution}

\newcommand{\etalchar}[1]{$^{#1}$}
\begin{thebibliography}{ADHW09}

\bibitem[ADHW09]{AbeyesingheDHW09}
Anura Abeyesinghe, Igor Devetak, Patrick Hayden, and Andreas Winter.
\newblock The mother of all protocols: {R}estructuring quantum information's
  family tree.
\newblock {\em Proceedings of the Royal Society of London}, A(465):2537--2563,
  2009.

\bibitem[BCF{\etalchar{+}}96]{barnum96}
Howard Barnum, Carlton~M. Cave, Christopher~A. Fuch, Richard Jozsa, and
  Benjamin Schmacher.
\newblock Noncommuting mixed states cannot be broadcast.
\newblock {\em Phys. Rev. Lett.}, 76(15):2818--2821, 1996.

\bibitem[BCR11]{Renner11}
Mario Berta, Matthias Christandl, and Renato Renner.
\newblock The {Q}uantum {R}everse {S}hannon {T}heorem based on one-shot
  information theory.
\newblock {\em Commun. Math. Phys.}, 306(3):579--615, 2011.

\bibitem[BCT16]{Berta14}
M.~Berta, M.~Christandl, and D.~Touchette.
\newblock Smooth entropy bounds on one-shot quantum state redistribution.
\newblock {\em IEEE Transactions on Information Theory}, 62(3):1425--1439,
  March 2016.

\bibitem[BD10]{BuscemiD10}
Francesco Buscemi and Nilanjana Datta.
\newblock The quantum capacity of channels with arbitrarily correlated noise.
\newblock {\em IEEE Transactions on Information Theory}, 56:1447--1460, 2010.

\bibitem[Ber09]{Berta09}
Mario Berta.
\newblock Single-shot quantum state merging.
\newblock Master's thesis, ETH Zurich, http://arxiv.org/abs/0912.4495, 2009.

\bibitem[BW92]{bennett92}
Charles~H. Bennett and Stephen~J. Wiesner.
\newblock Communication via one- and two-particle operators on
  einstein-podolsky-rosen states.
\newblock {\em Phys. Rev. Lett.}, 69(20):2881--2884, 1992.

\bibitem[Car10]{carlen}
Eric Carlen.
\newblock Trace inequalities and quantum entropy: an introductory course.
  entropy and the quantum.
\newblock {\em Contemp. Math.}, 529:73--140, 2010.

\bibitem[CHP07]{petz07}
I.~Csiszár, F.~Hiai, and D.~Petz.
\newblock Limit relation for quantum entropy and channel capacity per unit
  cost.
\newblock {\em J. Math. Phys.}, 48(092102), 2007.

\bibitem[CT91]{CoverT91}
Thomas~M. Cover and Joy~A. Thomas.
\newblock {\em Elements of information theory}.
\newblock Wiley Series in Telecommunications. John Wiley \& Sons, New York, NY,
  USA, 1991.

\bibitem[Dat09]{Datta09}
Nilanjana Datta.
\newblock Min- and max- relative entropies and a new entanglement monotone.
\newblock {\em IEEE Transactions on Information Theory}, 55:2816--2826, 2009.

\bibitem[DBWR14]{DupuisBWR14}
Fr{\'e}d{\'e}ric Dupuis, Mario Berta, J{\"u}rg Wullschleger, and Renato Renner.
\newblock One-shot decoupling.
\newblock {\em Communications in Mathematical Physics}, 328(251), 2014.

\bibitem[DCHR78]{CastelleHR78}
D.~Dacunha-Castelle, H.~Heyer, and B.~Roynette.
\newblock Ecole d'{E}t{\'e} de {P}robabilit{\'e}s de {S}aint-{F}lour {VII}.
\newblock {\em Lecture Notes in Mathematics, Springer-Verlag}, 678, 1978.

\bibitem[DH11]{DattaH11}
Nilanjana Datta and Min-Hsiu Hsieh.
\newblock The apex of the family tree of protocols: optimal rates and resource
  inequalities.
\newblock {\em New Journal of Physics}, 13(093042), 2011.

\bibitem[DHO14]{Oppenheim14}
Nilanjana Datta, Min-Hsiu Hsieh, and Jonathan Oppenheim.
\newblock An upper bound on the second order asymptotic expansion for the
  quantum communication cost of state redistribution.
\newblock http://arxiv.org/abs/1409.4352, 2014.

\bibitem[Dup10]{Frederic10}
Fr{\'e}d{\'e}ric Dupuis.
\newblock The decoupling approach to quantum information theory.
\newblock PhD Thesis, Universit{\'e} de Montr{\'e}al.,
  http://arxiv.org/abs/1410.0664, 2010.

\bibitem[DY08]{Devatakyard}
Igor Devetak and Jon Yard.
\newblock Exact cost of redistributing multipartite quantum states.
\newblock {\em Phys. Rev. Lett.}, 100(230501), 2008.

\bibitem[HHH99]{horodecki99}
Micha\l{} Horodecki, Pawe\l{} Horodecki, and Ryszard Horodecki.
\newblock General teleportation channel, singlet fraction, and
  quasidistillation.
\newblock {\em Phys. Rev. A}, 60:1888--1898, Sep 1999.

\bibitem[HJMR10]{HJMR10}
Prahladh Harsha, Rahul Jain, David Mc.Allester, and Jaikumar Radhakrishnan.
\newblock The communication complexity of correlation.
\newblock {\em IEEE Transcations on Information Theory}, 56:438--449, 2010.

\bibitem[HOW07]{horodecki07}
Micha{\l } Horodecki, Jonathan Oppenheim, and Andreas Winter.
\newblock Quantum state merging and negative information.
\newblock {\em Communications in Mathematical Physics}, 269(1):107--136, 2007.

\bibitem[IH08]{Portbased08}
Satoshi Ishizaka and Tohya Hiroshima.
\newblock Asymptotic teleportation scheme as a universal programmable quantum
  processor.
\newblock {\em Phys. Rev. Lett.}, 101:240501, Dec 2008.

\bibitem[IH09]{Portbased09}
Satoshi Ishizaka and Tohya Hiroshima.
\newblock Quantum teleportation scheme by selecting one of multiple output
  ports.
\newblock {\em Phys. Rev. A}, 79:042306, Apr 2009.

\bibitem[JRS03]{Jain:2003}
Rahul Jain, Jaikumar Radhakrishnan, and Pranab Sen.
\newblock A direct sum theorem in communication complexity via message
  compression.
\newblock In {\em Proceedings of the 30th international conference on Automata,
  languages and programming}, ICALP'03, pages 300--315, Berlin, Heidelberg,
  2003. Springer-Verlag.

\bibitem[JRS05]{Jain:2005}
Rahul Jain, Jaikumar Radhakrishnan, and Pranab Sen.
\newblock Prior entanglement, message compression and privacy in quantum
  communication.
\newblock In {\em Proceedings of the 20th Annual IEEE Conference on
  Computational Complexity}, pages 285--296, Washington, DC, USA, 2005. IEEE
  Computer Society.

\bibitem[LD09]{LuoD09}
Zhicheng Luo and Igor Devetak.
\newblock Channel simulation with quantum side information.
\newblock {\em IEEE Transactions on Information Theory}, 55:1331--1342, 2009.

\bibitem[Lin75]{lindblad75}
G.~Lindblad.
\newblock Completely positive maps and entropy inequalities.
\newblock {\em Commun. Math. Phys.}, 40:147--151, 1975.

\bibitem[MBD{\etalchar{+}}16]{Majenzetal}
Christian Majenz, Mario Berta, Fr{\'e}d{\'e}ric Dupuis, Renato Renner, and
  Matthias Christandl.
\newblock Catalytic decoupling of quantum information.
\newblock 2016.

\bibitem[NC00]{NielsenC00}
Michael~A. Nielsen and Isaac~L. Chuang.
\newblock {\em Quantum computation and quantum information}.
\newblock Cambridge University Press, Cambridge, UK, 2000.

\bibitem[Opp08]{oppenheim08}
Jonathan Oppenheim.
\newblock State redistribution as merging: introducing the coherent relay.
\newblock http://arxiv.org/abs/0805.1065, 2008.

\bibitem[Ren05]{Renner05}
Renato Renner.
\newblock Security of quantum key distribution.
\newblock PhD Thesis, ETH Zurich, Diss. ETH No. 16242, arXiv:quant-ph/0512258,
  2005.

\bibitem[Sch95]{Schumacher95}
Benjamin Schumacher.
\newblock Quantum coding.
\newblock {\em Phys. Rev. A.}, 51:2738--2747, 1995.

\bibitem[Sha]{Shannon}
Claude~Elwood Shannon.
\newblock A mathematical theory of communication.
\newblock {\em The Bell System Technical Journal}, 27:379--423.

\bibitem[Sti55]{stinespring55}
W.~F. Stinespring.
\newblock Positive functions on c*-algebras.
\newblock {\em Proceedings of the American Mathematical Society}, page
  211–216, 1955.

\bibitem[TCR10]{Tomamichel09}
Marco Tomamichel, Roger Colbeck, and Renato Renner.
\newblock Duality between smooth min- and max-entropies.
\newblock {\em IEEE Transactions on Information Theory}, 56(9):4674 -- 4681,
  2010.

\bibitem[Tom12]{Tomamichel12}
Marco Tomamichel.
\newblock A framework for non-asymptotic quantum information theory.
\newblock PhD Thesis, ETH Zurich, http://arXiv,org/abs/1203.2142, 2012.

\bibitem[Tou15]{Dave14}
Dave Touchette.
\newblock Quantum information complexity.
\newblock In {\em Proceedings of the Forty-Seventh Annual ACM on Symposium on
  Theory of Computing}, STOC '15, pages 317--326, New York, NY, USA, 2015. ACM.

\bibitem[Uhl76]{uhlmann76}
A.~Uhlmann.
\newblock The "transition probability" in the state space of a *-algebra.
\newblock {\em Rep. Math. Phys.}, 9:273--279, 1976.

\bibitem[Wat11]{Watrouslecturenote}
John Watrous.
\newblock Theory of {Q}uantum {I}nformation, lecture notes, 2011.
\newblock https://cs.uwaterloo.ca/~watrous/LectureNotes.html.

\bibitem[YBW08]{YeBW08}
Ming-Yong Ye, Yan-Kui Bai, and Z.~D. Wang.
\newblock Quantum state redistribution based on a generalized decoupling.
\newblock {\em Physical Review A}, 78(030302(R)), 2008.

\bibitem[YD09]{YardD09}
Jon~T. Yard and Igor Devetak.
\newblock Optimal quantum source coding with quantum side information at the
  encoder and decoder.
\newblock {\em IEEE Transactions on Information Theory}, 55:5339--5351, 2009.

\end{thebibliography}

\end{document}